\documentclass[a4paper,UKenglish,cleveref, autoref, thm-restate]{lipics-v2021}

\pdfoutput=1 
\hideLIPIcs  

\usepackage{mathtools}
\usepackage{tikz-cd}
\usepackage{xparse}
\usepackage{xspace}
\usepackage{mathpartir}
\usepackage{environ}
\usepackage{pict2e}
\usepackage{tikz}

\newtheorem{component}[theorem]{Component}
\newtheorem{notation}[theorem]{Notation}

\crefname{component}{Component}{Components}

\usetikzlibrary{arrows,shapes}

\DeclarePairedDelimiter\angles\langle\rangle
\DeclarePairedDelimiter\braces\lbrace\rbrace
\DeclarePairedDelimiterX\braceset[2]\lbrace\rbrace{#1 \mathrel{\delimsize\vert} #2}

\makeatletter
\newcommand{\DeclareAbbrevation}[2]{\newcommand{#1}{\@ifnextchar{.}{#2}{#2.\@\xspace}}}
\makeatother

\DeclareAbbrevation{\ie}{i.e}
\DeclareAbbrevation{\eg}{e.g}
\DeclareAbbrevation{\cf}{cf}
\DeclareAbbrevation{\etc}{etc}
\DeclareAbbrevation{\resp}{resp}
\DeclareAbbrevation{\etal}{et al}
\DeclareAbbrevation{\ibid}{ibid}
\DeclareAbbrevation{\ca}{ca}

\newcommand{\defeq}{\coloneqq}
\newcommand{\co}{\colon}

\tikzset{
  epi/.style = {commutative diagrams/two heads},
  mono/.style = {commutative diagrams/tail},
  cof/.style = {{Triangle[reversed,scale=0.9]}->},
  fib/.style = {-{Triangle[open,scale=0.9]}},
  weq/.style = {"\sim"'{sloped,font=\tiny,#1}},
  tcof/.style = {cof,weq={#1}},
  tcof/.default = {above},
  tfib/.style = {fib,weq={#1}},
  tfib/.default = {above},
  weq/.default = {above},
  bifib/.style = {{Triangle[reversed,scale=0.9]}-{Triangle[open,scale=0.9]}},
}

\tikzcdset{
  arrow style=tikz,
  phantomcenter/.style = {phantom,start anchor=center,end anchor=center}
}

\NewDocumentCommand{\pushout}{D<>{0} O{6ex}}{%
  \arrow[phantom,start anchor=center,to path={-- ++({#1+135}:#2) \tikztonodes}, "\ulcorner"]
}
\NewDocumentCommand{\pullback}{D<>{0} O{6ex}}{%
  \arrow[phantom,start anchor=center,to path={-- ++({#1-45}:#2) \tikztonodes}, "\lrcorner"]
}

\NewDocumentCommand{\id}{o}{\mathrm{id}\IfValueT{#1}{_{#1}}}
\NewDocumentCommand{\Id}{o}{\mathrm{Id}\IfValueT{#1}{_{#1}}}

\newcommand{\pair}[2]{(#1, #2)}

\newcommand{\Slice}[2]{{#1} / {#2}}
\newcommand{\Coslice}[2]{{#1} / {#2}}

\NewDocumentCommand{\mono}{s}{\IfBooleanTF{#1}{\leftarrowtail}{\rightarrowtail}}
\NewDocumentCommand{\epi}{s}{\IfBooleanTF{#1}{\twoheadleftarrow}{\twoheadrightarrow}}

\NewDocumentCommand\YoSymScaled{m}{
  \raisebox{#1em * \real{-.02}}{%
    \includegraphics[quiet=true,draft=false,height=#1em * \real{.67},keepaspectratio]{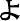}%
    \hspace{0.1em}%
  }
}

\NewDocumentCommand{\yosym}{}{
  \mathord{%
    \mathchoice{\YoSymScaled{1}}{\YoSymScaled{1}}{\YoSymScaled{\defaultscriptratio}}{\YoSymScaled{\defaultscriptscriptratio}}
  }%
}

\NewDocumentCommand{\PSh}{o m}{\mathrm{PSh}\IfValueT{#1}{_{#1}}(#2)}
\NewDocumentCommand{\yo}{g}{\yosym\IfValueT{#1}{#1}}

\newcommand{\subst}[1]{#1^*}

\newcommand{\ran}[1]{{#1}_*}

\NewDocumentCommand{\arrowofstyle}{m m m}{
  \NewDocumentCommand#1{o}{
    \mathrel{\hspace{-0.63ex}\tikz[baseline=-0.7ex, shorten <=3pt, shorten >=2pt]
      \draw[#2,line width=0.08ex](0,0) -- (1.6em,0)
      \IfValueT{##1}{node [midway,scale=0.75,yshift=0.5ex,xshift=#3] {$##1$}};
    \hspace{-0.41ex}}
  }
}

\arrowofstyle{\ordinary}{->}{-0.2ex}
\arrowofstyle{\ordinaryleft}{<-}{0.4ex}
\arrowofstyle{\cof}{cof}{0.0ex}
\arrowofstyle{\fib}{fib}{-0.3ex}
\NewDocumentCommand{\weq}{s}{\IfBooleanTF{#1}{\ordinaryleft[\sim]}{\ordinary[\sim]}}

\newcommand{\II}{\mathbb{I}}

\newcommand{\upsc}[1]{\textup{\textsc{#1}}}

\newcommand{\cat}[1]{\mathcal{#1}}
\renewcommand{\mod}[1]{\mathcal{#1}}
\newcommand{\rmc}[1]{\mathbb{#1}}
\newcommand{\RMC}{\mathbf{RMC}}
\newcommand{\rmcCl}[1]{\mathbb{C}\upsc{l}(#1)}
\newcommand{\rmcII}{\mathbb{I}\upsc{nt}}
\newcommand{\rmcDTT}{\mathbb{M}\upsc{ltt}}
\newcommand{\rmcCTT}{\mathbb{C}\upsc{tt}}
\newcommand{\rmcCTTs}{\mathbb{C}\upsc{tt}_s}
\newcommand{\rmcMLTT}{\mathbb{M}\upsc{ltt}_{\Sigma,\mathsf{Id}}}
\newcommand{\rmcMLTTU}{\mathbb{M}\upsc{ltt}_{\Sigma,\mathsf{Id},\ccU}}
\newcommand{\rmcCof}{\mathbb{C}\upsc{of}}
\newcommand{\rmcMLTTPi}{\mathbb{M}\upsc{ltt}_{\Sigma,\mathsf{Id},\Pi}}

\newcommand{\Mod}[1]{\mathbf{Mod}(#1)}
\newcommand{\Ctx}[1]{#1(\star)}
\newcommand{\heart}[1]{{#1}^{\heartsuit}}

\newcommand{\flipII}{\mathrm{Flip}}
\newcommand{\thRev}[2]{\mathrm{Rev}_{#2}#1}
\newcommand{\twistII}{\mathrm{T}}
\newcommand{\twistCTT}{\mathrm{T}}

\newcommand{\thCart}{\Phi_{\mathrm{cart}}}
\newcommand{\thDL}{\Phi_{\mathrm{DL}}}

\NewDocumentCommand{\spanCTT}{s m m}{\mathrm{S}\IfBooleanF{#1}{^{#2}_{#3}}}
\NewDocumentCommand{\apexCTT}{s m m}{\mathrm{M}\IfBooleanF{#1}{^{#2}_{#3}}}

\newcommand{\spanMLTT}{\mathrm{Refl}}
\newcommand{\apexMLTT}{\mathrm{P}}
\newcommand{\relMLTT}{\mathrm{R}}

\newcommand{\soO}{\square}
\newcommand{\soR}{\star}
\newcommand{\soto}{\Rightarrow}
\newcommand{\sov}[1]{\mathtt{#1}}
\newcommand{\soabs}[1]{\angles{#1}}

\newcommand{\ccTy}{\mathsf{Ty}}
\newcommand{\ccTm}{\mathsf{Tm}}
\newcommand{\ccII}{\mathbb{I}}
\newcommand{\ccCof}{\mathsf{Cof}}
\newcommand{\ccTrue}{\mathsf{True}}

\newcommand{\clTm}{\pi_{\ccTm}}
\newcommand{\clTrue}{\pi_{\ccTrue}}

\newcommand{\cce}[1]{\mathsf{#1}}

\newcommand{\ccTop}{\top}
\newcommand{\ccBot}{\bot}
\newcommand{\ccCup}{\cup}
\newcommand{\ccCap}{\cap}
\newcommand{\ccAbortTy}{\mathsf{elim}^{\ccTy}_\bot}
\newcommand{\ccAbortTm}{\mathsf{elim}^{\ccTm}_\bot}
\newcommand{\ccCaseTy}{\mathsf{elim}^{\ccTy}_\cup}
\newcommand{\ccCaseTm}{\mathsf{elim}^{\ccTm}_\cup}
\newcommand{\ccCofEq}{\approx}

\newcommand{\ccFill}{\mathsf{fill}}
\newcommand{\ccCoe}{\mathsf{coe}}

\newcommand{\ccSym}{\mathsf{sym}}

\newcommand{\ccSigma}{\Sigma}
\newcommand{\ccPair}{\mathsf{pair}}
\newcommand{\ccFst}{\mathsf{fst}}
\newcommand{\ccSnd}{\mathsf{snd}}
\newcommand{\ccTimes}{\mathbin{\times}}

\newcommand{\ccPi}{\Pi}
\newcommand{\ccLam}{\lambda}

\newcommand{\ccPath}{\mathsf{Path}}
\NewDocumentCommand{\ccPathLam}{s}{\IfBooleanTF{#1}{\lambda}{\lambda^\II}}
\newcommand{\ccPathApp}{\mathbin{\textup{\textsf{@}}}}

\newcommand{\ccId}{\mathsf{Id}}

\newcommand{\ccU}{\mathsf{U}}
\newcommand{\ccEl}{\mathsf{El}}

\newcommand{\ccSusp}{\mathsf{Susp}}
\newcommand{\ccNorth}{\mathsf{north}}
\newcommand{\ccSouth}{\mathsf{south}}
\newcommand{\ccMerid}{\mathsf{merid}}
\newcommand{\ccSuspElim}{\mathsf{elim}}
\newcommand{\ccSuspElimMerid}{\mathsf{merid\beta}}
\newcommand{\ccSuspMap}{\mathsf{map}}

\newcommand{\ccIsContr}{\mathsf{isContr}}
\newcommand{\ccIsEquiv}{\mathsf{isEquiv}}

\newcommand{\ccFwd}{\mathsf{fwd}}

\newcommand{\ccThicken}{\mathsf{thicken}}
\newcommand{\ccAnti}{\mathsf{anti}}

\newcommand{\ccCong}{\mathsf{cong}}

\newcommand{\ccDecode}[1]{\mathsf{decode}^{#1}}

\newcommand{\ccGraph}{\mathsf{Graph}}

\newcommand{\Span}[1]{\mathrm{Span}(#1)}
\newcommand{\spanl}{d^0}
\newcommand{\spanr}{d^1}
\newcommand{\R}{{\mathrm{r}}}

\newcommand{\ccTyEqv}{\mathsf{Ty}^{\simeq}}
\newcommand{\ccTmEqv}{\mathsf{Tm}^{\simeq}}
\newcommand{\ccTmId}{\mathsf{Tm}^{\asymp}}

\newcommand{\ccGlue}{\mathsf{Glue}}
\newcommand{\ccEnglue}{\mathsf{glue}}
\newcommand{\ccUnglue}{\mathsf{unglue}}

\newcommand{\unitl}{\eta^0}
\newcommand{\unitr}{\eta^1}

\newcommand{\Synt}[1]{\mathbf{0}_{#1}}

\newcommand{\Omegadec}{\Omega_{\mathrm{dec}}}
\newcommand{\OPCof}{\Omega_{\mathrm{cof}}}

\newcommand{\gencof}[2]{#1 \otimes_{#2} \delta}

\title{Eliminating reversals from cubical type theories}

\author{Evan Cavallo}{Department of Computer Science and Engineering, University of Gothenburg and Chalmers University of Technology, Sweden \and \url{https://ecavallo.net} }{evan.cavallo@gu.se}{https://orcid.org/0000-0001-8174-7496}{Knut and Alice Wallenberg Foundation (KAW), Grant No.\ 2019.0116}

\author{Christian Sattler}{Department of Computer Science and Engineering, Chalmers University of Technology and University of Gothenburg, Sweden \and \url{https://www.cse.chalmers.se/~sattler/}}{sattler@chalmers.se}{0000-0001-6374-4427}{US Air Force Office of Scientific Research, award number FA9550-24-1-0302}

\authorrunning{E. Cavallo and C. Sattler} 

\Copyright{Evan Cavallo and Christian Sattler} 



\category{} 

\relatedversion{} 

\nolinenumbers 

\EventEditors{Claudia Faggian and Joost-Pieter Katoen}
\EventNoEds{2}
\EventLongTitle{41st Annual Symposium on Logic in Computer Science (LICS 2026)}
\EventShortTitle{LICS 2026}
\EventAcronym{LICS}
\EventYear{2026}
\EventDate{July 20--23, 2026}
\EventLocation{Lisbon, Portugal}
\EventLogo{}
\SeriesVolume{380}
\ArticleNo{82}

\begin{document}

\maketitle

\begin{abstract}
  Cubical type theories are designed around an abstract unit interval from which types of paths, used to represent equalities, are defined.
  Varying the operations available on this interval yields different type theories.
  A reversal is an involutive operator on the interval that swaps its two endpoints.
  We show that for cubical type theories with self-dual interval theories, such as the minimal theory of two endpoints or the theory of a bounded distributive lattice, the extension of the theory with a reversal that internalizes the duality is a conservative extension.
  The key tool is a ``twist construction'': the product of an interval and its dual is again an interval with a reversal given by swapping coordinates.

  Our conservativity result applies to ``opaque'' cubical type theories, without strict equations reducing the filling operator at concrete type formers or eliminators from higher inductive types at path constructors.
  Using the same twist construction, we also construct models of strict cubical type theory with reversals in categories of cubical sets without reversals.
  We thereby give the first model of a theory with reversals whose homotopy theory corresponds to that of topological spaces.
\end{abstract}

\section{Introduction}

Cubical type theories \cite{cohen-coquand-huber-mortberg:15,angiuli-favonia-harper:18,abcfhl:21} extend Martin-L\"of's dependent type theory \cite{martin-lof:75} with an abstract unit interval $\II$ which behaves much like a type.
Types of \emph{paths} $a_0 \sim^A a_1$, \ie, of terms $i : \II \vdash a(i) : A$ varying over the interval with fixed endpoints $a(0) = a_0$ and $a(1) = a_1$, play the role of equality types.
As equality types, path types are remarkably well-behaved.
For example, they natively satisfy function extensionality: equalities of functions correspond to families of pointwise equalities.
With additional type formers, cubical type theories can also support Voevodsky's univalence axiom and higher inductive types (HITs) \cite{coquand-huber-mortberg:18,cavallo-harper:19}, making them models of homotopy type theory (HoTT) \cite{hott-book}.

Path types satisfy different strict equations than Martin-L\"of's identity types.
On the one hand, they do not support a J eliminator with a strict computation rule \cite[\S9.1]{cohen-coquand-huber-mortberg:15}.
On the other hand, for example, one has an operator witnessing that functions $f : A \to B$ preserve paths,
$\ccCong_f \defeq \ccLam p. \ccPathLam* i. f(p(i)): (a_0 \sim^A a_1) \to (f(a_0) \sim^B f(a_1))$,
that commutes \emph{strictly} with function composition: $\ccCong_g \circ \ccCong_f = \ccCong_{g \circ f}$.
Such equations make cubical type theory a convenient setting for \emph{synthetic homotopy theory} (see, \eg, M\"ortberg and Pujet \cite{mortberg-pujet:20}), homotopy theory developed in the language of type theory, which can involve complicated manipulations with iterated identity/path types.

The range of strict equations satisfied by a cubical type theory's path types depends on its \emph{interval theory}, the collection of operations available on $\II$.
Given a \emph{reversal} operator $i : \II \vdash \lnot i : \II$ such that $\lnot 0 = 1$, $\lnot 1 = 0$, and $\lnot \lnot i = i$, we can define a path inversion operator
\[
  \ccSym \defeq \ccLam p. \ccPathLam* i. p(\lnot i) : (a_0 \sim^A a_1) \to (a_1 \sim^A a_0)
\]
that is strictly involutive ($\ccSym \circ \ccSym = \id$) and commutes strictly with the action of functions ($\ccCong_f \circ \ccSym = \ccSym \circ \ccCong_f$).
\emph{Connections} $i : \II, j : \II \vdash i \land j : \II$ and $i : \II, j : \II \vdash i \lor j : \II$ behaving like the $\min$ and $\max$ functions on the topological interval are similarly useful for higher-dimensional manipulations.
Cohen, Coquand, Huber, and M\"ortberg's original cubical type theory \cite{cohen-coquand-huber-mortberg:15} includes $\lnot$, $\land$, and $\lor$ with the equational theory of the free De Morgan algebra.
On the other hand, Angiuli, Favonia, and Harper's theory \cite{angiuli-favonia-harper:18} demonstrates that none of these operators is \emph{necessary} to set up a well-behaved cubical type theory.

While convenient for the user of the type theory, additional operations on the interval are less convenient for the semanticist.
To justify the project of synthetic homotopy theory, a cubical type theory should at least have a model in \emph{$\infty$-groupoids}, an abstract description of the homotopy theory of topological spaces.
Constructive models classically equivalent to $\infty$-groupoids were found first for cubical type theory without any interval operations by Awodey, Cavallo, Coquand, Riehl, and Sattler~\cite{accrs:24} and then for the theory with one connection $\lor$ by Cavallo and Sattler~\cite{cavallo-sattler:25}.
Most recently, the second-named author announced \cite{sattler:25} a model constructively equivalent to $\infty$-groupoids that can interpret cubical type theory with two connections, $\land$ and $\lor$, and the equations of a bounded distributive lattice.
However, none of these models interpret a reversal.
This is a particularly unfortunate state of affairs because Cubical Agda \cite{vezzosi-mortberg-abel:21}, the most widely used proof assistant for cubical type theory, is based on Cohen et al.'s type theory and thus includes $\lnot$ along with $\lor$ and $\land$, and its substantial standard library \cite{agda-cubical} relies extensively on these operators.

\subsection{Contributions}
\label{sec:contributions}

We show that a reversal is an essentially harmless extension to cubical type theory.

The key fact is that when $\II$ is an interval object with endpoints $0$ and $1$, its square $\II \times \II$ is an interval object with endpoints $(0,1)$ and $(1,0)$ and a reversal $\lnot (i_0,i_1) \defeq (i_1,i_0)$ that swaps the axes of the square.
When $\II$ has connections defining a distributive lattice $(\II,\land,\lor,0,1)$, $\II \times \II$ is a De Morgan algebra with connections given by $(i_0,i_1) \land (j_0,j_1) \defeq (i_0 \land j_0, i_1 \lor j_1)$ and $(i_0,i_1) \lor (j_0,j_1) \defeq (i_0 \lor j_0, i_1 \land j_1)$.
In general, when $\II$ has some self-dual algebraic structure (in a sense we make precise in \cref{sec:extension-by-reversal}), $\II \times \II$ has the same structure as well as a reversal.
A variety of constructions in this mold appear in the algebraic literature (e.g., in lattice and order theory), where they are called \emph{twist constructions}.
This name originates with Kracht \cite{kracht:98}, who applies it to a construction of Nelson algebras from Heyting algebras taken from Vakarelov \cite{vakarelov:77}.
Fidel and Brignole \cite{fidel-brignole:91} and Rivieccio \cite[\S7]{rivieccio:20} consider the case of building De Morgan algebras from distributive lattices, which is of particular interest to us.

We derive two main results from this simple construction.

\subsubsection{Conservativity for opaque cubical type theory}

First, we prove that a reversal is a conservative extension for ``opaque'' cubical type theories with self-dual interval theories.
Similar to a theory considered by Coquand, Huber, and Sattler \cite{coquand-huber-sattler:22}, these opaque theories are cubical type theories where certain strict equations are either omitted or replaced with terms of path type.
Specifically, we
\begin{enumerate}[(a)]
\item\label{conservativity-omit-equations} omit equations that reduce uses of the filling operator at concrete type formers, and
\item weaken equations for the reduction of HIT eliminators on path constructors to paths.
\end{enumerate}
The equations of \eqref{conservativity-omit-equations} always hold up to paths, by higher dimensional-instances of filling, so omitting them also amounts to weakening them to paths.

Building on the twist construction, we define a \emph{twist interpretation} from opaque cubical type theories with reversals to corresponding theories without reversals, interpreting judgments $\Gamma \vdash i : \II$ as pairs $(\Gamma \vdash i_0 : \II, \Gamma \vdash i_1 : \II)$ and path types as \emph{square} types---encoded as iterated path types---with fixed values at the points $(0,1)$ and $(1,0)$.
We use this translation to prove a conservativity result: for a context of term variables $\Gamma$ and type $\Gamma \vdash A$ in the base theory and a term $\Gamma \vdash_{\lnot} N : A$ in the extended theory, there is a term $\Gamma \vdash M : A$ in the base theory with a path $\Gamma \vdash_\lnot P : M \sim^A N$.
Similarly, if $\Gamma \vdash_\lnot B$ is a type in the extended theory, then there is a type $\Gamma \vdash A$ with an equivalence $\Gamma \vdash_\lnot E : A \simeq B$.

Coquand, Huber, and Sattler \cite{coquand-huber-sattler:22} prove \emph{homotopy canonicity} for their opaque cubical type theory: every closed term of natural number type is connected by a path to a concrete numeral.
Our hope is that similar techniques can be used to show that strict cubical type theories may be conservative over their opaque counterparts in general.
Some progress towards a framework for coherence theorems dealing with strictification of equations has been made by Bocquet \cite{bocquet:23}.
With such a result, the program of relating cubical theories with different interval structure could be conducted in the simpler world of opaque theories, as we do here for reversals, then mechanically extended to strict theories.

\subsubsection{Models for strict cubical type theory with reversals in spaces}

In lieu of a hoped-for conservativity result for strict over opaque cubical type theories, we separately use the same basic twist construction to build concrete models of strict cubical type theory with reversals.
We work with the parameterized model construction of Angiuli, Brunerie, Coquand, Harper, Favonia, and Licata (ABCHFL) \cite{abcfhl:21}, showing that any model of cubical type theory given by this construction can be upgraded to a model of a cubical type theory with reversals in the same target category.

Combining this with our prior work showing that the homotopy theory of a certain ABCHFL model is classically equivalent to $\infty$-groupoids \cite{cavallo-sattler:25}, we obtain a model of strict cubical type theory with a reversal whose homotopy theory is classically equivalent to $\infty$-groupoids.
This is the first known model of its kind.
Although Cohen, Coquand, Huber, and M\"ortberg~\cite{cohen-coquand-huber-mortberg:15} give a model of their type theory, which includes reversals, this model and others like it have pathological homotopy theories \cite{sattler:18}.
Our models avoid these pathologies.

To obtain a similar model for strict cubical type theory with reversals \emph{and} connections, we would need an input ABCHFL model of the theory with connections in $\infty$-groupoids.
At present, no such model is known: while the theory with connections has an ABCHFL model in cartesian cubical sets with connections, it is an open problem to characterize its homotopy theory, as discussed by Streicher and Weinberger \cite{streicher-weinberger:21}.
The second-named author has claimed a model of cubical type theory with connections whose homotopy theory is that of $\infty$-groupoids \cite{sattler:25}, but it is not a direct instance of the ABCHFL construction.
We expect that our work adapts to this model, but we leave this for future work.

\subsection{Outline}

In \S\ref{sec:tt}, we set the stage to study cubical type theories in generality by reviewing Uemura's framework of \emph{second-order generalized algebraic theories} and their semantics based on \emph{representable map categories} \cite{uemura:21,uemura:23}.
We recall Kapulkin and Lumsdaine's definition of \emph{weak equivalence} for models of type theory with $\Sigma$ and identity types \cite{kapulkin-lumsdaine:18}, which we will use to state our conservativity theorem.
In \S\ref{sec:cubical}, we define cubical type theory in these terms.

We begin studying opaque cubical type theory in \S\ref{sec:twist}, where we define the extension of a self-dual interval theory with a reversal and the \emph{twist interpretation} from a cubical type theory extended with a reversal back to the original theory.
In \S\ref{sec:span}, we define the \emph{representable map category of spans} and develop tools to relate pairs of interpretations between cubical type theories.
We use these in \S\ref{sec:conservativity} to prove a general theorem for deriving weak equivalences from interpretations; we apply it with the twist interpretation to deliver the conservativity of reversals over opaque cubical type theories with self-dual interval theories.
In \S\ref{sec:spaces}, we construct a model of strict cubical type theory with reversals whose homotopy theory is classically that of $\infty$-groupoids.

\section{Type theories and models}
\label{sec:tt}

\subsection{SOGATs}
\label{sec:sogats}

We use Uemura's framework of \emph{second-order generalized algebraic theories (SOGATs)}~\cite[Chapter 4]{uemura:21} and their semantic counterparts, \emph{representable map categories}~\cite[Chapter 4]{uemura:21}~\cite{uemura:23} (also called \emph{categories with representable maps}).
The language of SOGATs is a \emph{logical framework} (\cf Harper, Honsell, and Plotkin~\cite{harper-honsell-plotkin:93}) with which we can specify type theories themselves in type-theoretic language.
In a SOGAT, we specify the judgment forms of a type theory while marking some as hypothesizable or \emph{representable}.
To begin with an example, \emph{basic dependent type theory} $\rmcDTT$~\cite[Example 4.6.1]{uemura:21} is specified by two sorts, one for types and one for terms:
\[
  \begin{array}{lcl}
    \ccTy &:& () \soto \soO
  \end{array}
  \qquad\qquad
  \begin{array}{lcl}
    \ccTm &:& (\sov{A} : () \to \ccTy) \soto \soR
  \end{array}
\]
The type sort $\ccTy$ takes no parameters ($() \soto \ldots$) and is \emph{non-representable} ($\soO$), so hypotheses ``$A : \ccTy$'' cannot appear in the contexts of the type theory being specified.
The term sort $\ccTm$ takes one type $(\sov{A} : () \to \ccTy)$ as a parameter and is \emph{representable} ($\soR$), meaning we allow hypotheses ``$a : \ccTm(A)$''.

In general, a SOGAT consists of declarations $\Phi \soto s$ where $\Phi$ is an \emph{environment} and $s$ is either $\soO$, $\soR$, a previously introduced sort, or an equation between expressions such a sort.
An environment is a list $(\sov{A}_1 : \Gamma_1 \to e_1, \ldots, \sov{A}_n : \Gamma_n \to e_n)$ of metavariables $\sov{A}_i$ that take a \emph{context} $\Gamma_i$ and output an $e_i$ which is either a previously defined sort or an equation between expressions in such a sort.
Finally, a \emph{context} is a list $(\sov{a}_1 : A_1, \ldots, \sov{a}_n : A_n)$ of variables $\sov{a}_i$ of representable sorts $A_i$.
Everything can depend on what precedes it.
We refer to Uemura~\cite[Chapter 4]{uemura:21} for a much more precise description.

\begin{notation}
  We omit empty contexts $(() \to \ldots)$ and environments $(() \soto \ldots)$, writing, \eg, $\ccTy : \soO$ and $\ccTm : (\sov{A} : \ccTy) \soto \soR$.
  We also omit variable names when what follows does not depend on them, as in $\ccTm : \ccTy \soto \soR$.
  We surround arguments in square brackets when we intend to leave them implicit, as in the following example of $\Sigma$ types.
  Specifically to $\rmcDTT$: we leave the $\ccTm$ operator implicit and write $\sov{a} : \sov{A}$ rather than $\sov{a} : \ccTm(\sov{A})$.
\end{notation}

\begin{notation}[{{cf.~\cite[Remark 5.4.6]{uemura:21}}}]
  Any context can be treated as an environment; one level up, any environment $\Phi$ over a SOGAT $T$ can be treated as an extension of $T$ with new declarations.
  We write $T[\Phi]$ for the extended SOGAT combining $T$ with $\Phi$.
\end{notation}

Uemura gives encodings of standard type formers of Martin-L\"of type theory such as (negative) unit types, (negative) dependent sums, and dependent products~\cite[\S4.6.1]{uemura:21}.
For example, dependent sums are specified by the declarations
\[
  \begin{array}{lcl}
    \ccSigma &:& (\sov{A} : \ccTy, \sov{B} : \sov{A} \to \ccTy) \soto \ccTy \\
    \ccFst &:& ([\sov{A} : \ccTy, \sov{B} : \sov{A} \to \ccTy], \ccSigma(\sov{A}, \sov{B})) \soto \sov{A} \\
    \ccSnd &:& ([\sov{A} : \ccTy, \sov{B} : \sov{A} \to \ccTy], \ccSigma(\sov{A}, \sov{B})) \soto \sov{B}(\sov{a}) \\
    \ccPair &:& ([\sov{A} : \ccTy, \sov{B} : \sov{A} \to \ccTy], \sov{a} : \sov{A}, \sov{b} : \sov{B}(\sov{a})) \soto \ccSigma(\sov{A}, \sov{B})
  \end{array}
\]
and, over $\Phi_\ccSigma = (\sov{A} : \ccTy, \sov{B} : \sov{A} \to \ccTy)$, equations
\[
\begin{array}{lcl}
    \_ &:& (\Phi_\ccSigma, \sov{a} : \sov{A}, \sov{b} : \sov{B}(\sov{a})) \soto \ccFst(\ccPair(\sov{a},\sov{b})) \equiv \sov{a} : \sov{A} \\
    \_ &:& (\Phi_\ccSigma, \sov{a} : \sov{A}, \sov{b} : \sov{B}(\sov{a})) \soto \ccSnd(\ccPair(\sov{a},\sov{b})) \equiv \sov{b} : \sov{B}(\sov{a}) \\
    \_ &:& (\Phi_\ccSigma, \sov{s} : \ccSigma(\sov{A},\sov{B})) \soto \sov{s} \equiv \ccPair(\ccFst(\sov{s}),\ccSnd(\sov{s})) : \ccSigma(\sov{A},\sov{B})
  \end{array}
\]

\begin{notation}
  We write $\rmcMLTT$ for the extension of $\rmcDTT$ with $\Sigma$ types, unit types (which we think of as nullary $\Sigma$ types), and identity types~\cite[Examples 4.6.4--4.6.6]{uemura:21}.
  We write $\rmcMLTTPi$ for its further extension with $\Pi$ types \cite[Example 4.6.3]{uemura:21}.
\end{notation}

\begin{notation}
  We write $\ccSigma \sov{a}{:}A.\ B$ as shorthand for $\ccSigma(A, \soabs{\sov{a}} B)$, where $\soabs{\sov{a}}$ denotes abstraction over the variable $\sov{a}$.
  If $B$ does not depend on $\sov{a}$, we write $A \ccTimes B$.
  We write $(a, b)$ for the pairing $\ccPair(a, b)$ and $s.1$ and $s.2$ for $\ccFst(s)$ and $\ccSnd(s)$, respectively.
  For $\Pi$ types, we write $\ccPi \sov{a}{:}A.\ B$ for $\ccPi(A, \soabs{\sov{a}} B)$ and $A \to B$ when $B$ does not depend on $\sov{a}$.
  For identity types, we write the type $\ccId({A},{a}_0,{a}_1)$ of identities in $A$ from $a_0$ to $a_1$ as ${a}_0 \asymp^{{A}} {a}_1$ or ${a}_0 \asymp {a}_1$.
\end{notation}

\begin{notation}
  The unit type and dependent sums justify types of dependent $n$-tuples for $n \geq 0$.
  We write these with tupling $(a_1, \ldots, a_n)$ and projections $s.1, \ldots, s.n$.
\end{notation}

\subsection{Representable map categories}

To specify the notion of \emph{model} of a SOGAT, Uemura first introduces \emph{representable map categories}, also called \emph{categories with representable maps}.

\begin{definition}[{Uemura \cite[Definition 3.2.1]{uemura:21}}]
  A \emph{representable map category (RMC)} is a finite limit category $\rmc{R}$ equipped with a class of morphisms, the \emph{representable maps}, such that
  \begin{enumerate}[(a)]
  \item the representable maps are closed under pullback, and
  \item for each representable map $f \colon Y \to X$, the pullback functor $\subst{f} \colon \Slice{\rmc{R}}{X} \to \Slice{\rmc{R}}{Y}$ has a right adjoint $f_* \colon \Slice{\rmc{R}}{Y} \to \Slice{\rmc{R}}{X}$ (called \emph{pushforward}).
  \end{enumerate}
  We use the arrow style $\fib$ to indicate representable maps.
  A \emph{representable map functor} or \emph{RMC functor} $F \colon \rmc{R} \to \rmc{S}$ between representable map categories is a functor that preserves finite limits, representable maps, and pushforwards along representable maps.
\end{definition}

\begin{example}[{Uemura \cite[Example 3.2.2]{uemura:21}}]
  \label{presheaf-rmc}
  Let $\cat{C}$ be a small category.
  The category of presheaves $\PSh{\cat{C}}$ becomes an RMC when equipped with the class of morphisms $f \colon B \to A$ such that for every map $a \colon \yo c \to A$ from a representable presheaf, there is a pullback square
  \[
    \begin{tikzcd}
      \yo d \pullback \ar[dashed]{d} \ar[dashed]{r} & B \ar{d}{f} \\
      \yo c \ar{r}[below]{a} & A
    \end{tikzcd}
  \]
  for some $d \in \cat{C}$.
\end{example}

The collection of representable map categories, representable map functors between them, and natural isomorphisms defines a (2,1)-category $\RMC$.
Each SOGAT $T$ induces a ``\emph{syntactic}'' RMC $\rmcCl{T}$ \cite[\S4.8]{uemura:21} whose objects are environments $\Phi$ over $T$ and whose morphisms are instantiations: an \emph{instantiation} $I \colon \Phi \to \Psi$ where $\Psi = (\sov{A}_1 : \Gamma_1 \to e_1, \ldots, \sov{A}_n : \Gamma_n \to e_n)$ is an assignment $(\sov{A}_1 \defeq \soabs{\vec{\sov{a}_1}}t_1, \ldots, \sov{A}_n \defeq \soabs{\vec{\sov{a}_n}}t_n)$ sending each metavariable $\sov{A}_i$ in the target to an expression $t_i : e_i$ in context $\vec{\sov{a}_i} : \Gamma_i$ over $T[\Phi]$.
An instantiation is representable when it is isomorphic to the projection $\Phi. \Gamma \to \Phi$ for an extension of an environment $\Phi$ by a context $\Gamma$.
For concrete SOGATs, we usually suppress $\rmcCl{-}$ and use the same name for the SOGAT and its induced RMC.

The RMC $\rmcCl{T}$ has a (2, 1)-categorical universal property that characterizes RMC functors $\rmcCl{T} \to \rmc{R}$ up to isomorphism as \emph{interpretations} \cite[Theorem 4.8.18]{uemura:21}.
An interpretation of $T$ in $\rmc{S}$ is a specification of the image of each declaration of $T$ inside $\rmc{S}$.
For example, an RMC functor $F \colon \rmcDTT \to \rmc{S}$ is determined up to isomorphism by an object $F\ccTy \in \rmc{S}$ and a representable map $F\clTm \colon F\ccTm \fib F\ccTy$, which specifies the image of $\clTm \colon (\sov{A} : \ccTy, \sov{a} : \ccTm(\sov{A})) \to (\sov{A} : \ccTy)$.
As a special case, we can speak of interpretations of a SOGAT $T$ in another SOGAT $S$ as interpretations of $T$ in $\rmcCl{S}$.

\subsection{Models}

An interpretation $\rmcCl{T} \to \rmc{S}$ is a model of a SOGAT as a \emph{second-order} theory.
To recover a notion of \emph{first-order} model, corresponding for example to categories with families \cite{dybjer:96} for $\rmcDTT$, Uemura uses presheaf categories with representable maps:

\begin{definition}[{{\cite[\S3.2.4]{uemura:21}}}]
  A \emph{model} $\mod{M} = (\cat{C},M)$ of an RMC $\rmc{R}$ is a category $\cat{C}$ with a terminal object and an RMC functor $M \colon \rmc{R} \to \PSh{\cat{C}}$ to the presheaf RMC of \cref{presheaf-rmc}.
  We write $\Ctx{\mod{M}}$ for $\cat{C}$ and $\mod{M}(X) \defeq MX \in \PSh{\Ctx{\mod{M}}}$ for $X \in \rmc{R}$.

  A \emph{morphism} $\mod{F} = (F,\alpha) \colon \mod{M} \to \mod{N}$ between models is a functor $F \colon \Ctx{\mod{M}} \to \Ctx{\mod{N}}$ and family of natural transformations $\alpha_X \colon \mod{M}(X) \to \subst{F}\mod{N}(X)$, natural in $X \in \rmc{R}$, such that for each representable $f \colon Y \fib X$, the naturality square for $\alpha$ at $f$ satisfies a Beck--Chevalley condition.
  For $c \in \Ctx{\mod{M}}$, we write $\mod{F}(c) \in \Ctx{\mod{N}}$ for $Fc$.
  For $x \colon \yo c \to \mod{M}(X)$ in $\PSh{\cat{C}}$, we write $\mod{F}_X(x) \colon \yo \mod{F}(c) \to \mod{N}(X)$ for the map corresponding by Yoneda to $\alpha_X \circ x \colon \yo c \to \subst{F}\mod{N}(X)$.
\end{definition}

Models of $\rmcDTT$ in this sense correspond directly to natural models as defined by Awodey \cite{awodey:18}, and thereby to categories with families: $\Ctx{\mod{M}}$ interprets the context judgment, and $\mod{M}(\ccTy)$ and $\mod{M}(\ccTm)$ the type and term judgments over a context.
With an appropriate notion of 2-morphism, the collection of models of an RMC $\rmc{R}$ forms a (2,1)-category $\Mod{\rmc{R}}$.

\begin{definition}[{\cite[Definitions 5.1.4 \& 5.1.6]{uemura:21}}]
  \label{democracy}
  The class of \emph{contextual objects} in $\Ctx{\mod{M}}$ for a model $\mod{M} \in \Mod{\rmc{R}}$ is inductively generated as follows:
  \begin{enumerate}
  \item terminal objects $1 \in \Ctx{\mod{M}}$ are contextual;
  \item for each contextual $c \in \Ctx{\mod{M}}$, representable $f \colon Y \fib X$ in $\rmc{R}$, and pullback square
    \[
      \begin{tikzcd}[sep=small]
        \yo{d} \pullback \ar{d} \ar{r} & \mod{M}(Y) \ar[fib]{d}{\mod{M}(f)} \\
        \yo{c} \ar{r} & \mod{M}(X)
      \end{tikzcd}
    \]
    with $d \in \Ctx{\mod{M}}$, the object $d$ is contextual.
  \end{enumerate}
  A model is \emph{democratic} when all of its objects are contextual.
  The \emph{heart} (or \emph{contextual core}) $\heart{\mod{M}} \in \Mod{\rmc{R}}$ of $\mod{M}$ is defined by taking $\Ctx{\heart{\mod{M}}}$ to be the full subcategory of contextual objects in $\Ctx{\mod{M}}$ and $\heart{\mod{M}}(X)$ to be the restriction of $\mod{{M}}(X)$ to a presheaf on $\Ctx{\heart{\mod{M}}}$.
\end{definition}

\subsection{Weak equivalences}

Kapulkin and Lumsdaine \cite[Deﬁnition 3.1]{kapulkin-lumsdaine:18} define \emph{weak equivalences} of contextual categories with identity types.
We translate their definition into Uemura's framework as a property of morphisms in $\Mod{\rmcMLTT}$.
First, we define the environment $\ccTyEqv$ of \emph{1-to-1 correspondences}, pairs of types connected by a type-valued relation that associates each element of one type with a unique element of the other.
This is one way of defining \emph{equivalence} between types \cite[Exercise 4.2]{hott-book}.
Similarly, we have an environment $\ccTmId$ of pairs of identified elements within a type.

\begin{definition}
  \label{correspondence}
  Over $\sov{A} : \ccTy$, define $\Phi_\ccIsContr(\sov{A}) \defeq (\sov{a}_0{:} \sov{A}, \sov{p} : (\sov{a}_1{:}\sov{A}) \to \sov{a}_0 \asymp^{\sov{A}} \sov{a}_1)$.
  Write $\ccTyEqv \in \rmcMLTT$ for
  \[
  \begin{array}{l}
    (\sov{A} : \ccTy, \sov{A}' : \ccTy, \overline{\sov{A}} : (\sov{a} : \sov{A}, \sov{a}' : \sov{A}') \to \ccTy, \\
    \phantom{(}{\_} : (\sov{a} : \sov{A}) \to \Phi_{\ccIsContr}(\ccSigma \sov{a}'{:}\sov{A}'.\ \overline{\sov{A}}(\sov{a},\sov{a}')),\ {\_} : (\sov{a}' : \sov{A}') \to \Phi_{\ccIsContr}(\ccSigma \sov{a}{:}\sov{A}.\ \overline{\sov{A}}(\sov{a},\sov{a}')))
    \end{array}
  \]
  and $\spanl, \spanr \colon \ccTyEqv \to \ccTy$ for the maps projecting $\sov{A}$ and $\sov{A}'$ respectively.
\end{definition}

\begin{definition}
  Set $\ccTmId \defeq (\sov{A} : \ccTy, \sov{a} : \sov{A}, \sov{a}' : \sov{A}, \overline{\sov{a}} : \sov{a} \asymp^{\sov{A}} \sov{a}')$ and write $\spanl, \spanr \colon \ccTmId \to \ccTm$ for the maps projecting $(\sov{A},\sov{a})$ and $(\sov{A},\sov{a}')$.
\end{definition}

\begin{definition}
  \label{weak-equivalence}
  A morphism $\mod{F} \colon \mod{M} \to \mod{N}$ in $\Mod{\rmcMLTT}$ is a \emph{weak equivalence} if the following hold for all $\Gamma \in \Ctx{\mod{M}}$:
  \begin{enumerate}[(a)]
  \item weak type lifting: for every $B \colon \yo \mod{F}(\Gamma) \to \mod{N}(\ccTy)$, there exist $A \colon \yo \Gamma \to \mod{M}(\ccTy)$ and $E \colon \yo \mod{F}(\Gamma) \to \mod{N}(\ccTyEqv)$ fitting in a commutative diagram
    \[
      \begin{tikzcd}[column sep=large]
        & \yo \mod{F}(\Gamma) \ar{dl}[above left]{\mod{F}_\ccTy(A)}  \ar{d}{E} \ar{dr}{B} \\
        \mod{N}(\ccTy) & \ar{l}[below]{\mod{N}(\spanl)} \mod{N}(\ccTyEqv) \ar{r}[below]{\mod{N}(\spanr)} & \mod{N}(\ccTy) \rlap{.}
      \end{tikzcd}
    \]
  \item weak term lifting: for every $A \colon \yo \mod{F}(\Gamma) \to \mod{N}(\ccTy)$ and $b \colon \yo \Gamma \to \mod{M}(\ccTm)$ with $\clTm b = \mod{F}_\ccTy(A)$,
    there exist $a \colon \yo \Gamma \to \mod{M}(\ccTm)$ with $\clTm a = A$ and $p \colon \yo \mod{F}(\Gamma) \to \mod{N}(\ccTmId)$ fitting in a commutative diagram
    \[
      \begin{tikzcd}[column sep=large]
        & \yo \mod{F}(\Gamma) \ar{dl}[above left]{\mod{F}_\ccTm(a)}  \ar{d}{p} \ar{dr}{b} \\
        \mod{N}(\ccTm) & \ar{l}[below]{\mod{N}(\spanl)} \mod{N}(\ccTmId) \ar{r}[below]{\mod{N}(\spanr)} & \mod{N}(\ccTm) \rlap{.}
      \end{tikzcd}
    \]
  \end{enumerate}
\end{definition}

Though we state \Cref{weak-equivalence} for arbitrary models, it is generally only well-behaved for democratic models.
In informal turnstile notation, $\mod{F} \colon \mod{M} \to \mod{N}$ is a weak equivalence when
\begin{enumerate}[(a)]
\item for every type $\mod{F}(\Gamma) \vdash_{\mod{N}} B$ in the target model, there is a type $\Gamma \vdash_{\mod{M}} A$ in the source model whose image by $\mod{F}$ is equivalent to $B$, and 
\item for every term $\mod{F}(\Gamma) \vdash_{\mod{N}} b : \mod{F}_\ccTy(A)$ in the target model, there is a term $\Gamma \vdash_{\mod{M}} a : A$ whose image by $\mod{F}$ is identified with $b$.
\end{enumerate}

We apply the notion of weak equivalence to ``syntactic'' models of $\rmcMLTT$, coming from extensions of the SOGAT of $\rmcMLTT$, in order to speak about conservativity relations between type theories (\cf for example Isaev \cite{isaev:20}, Bocquet \cite{bocquet:23}, Kapulkin and Li \cite{kapulkin-li:25}).

\begin{definition}
  \label{models-from-syntax}
  For an RMC functor $F \colon \rmc{R} \to \rmc{S}$, we write $\Synt{F} \defeq \heart{(\rmc{R},\yo \circ F)} \in \Mod{\rmc{R}}$ for the heart of the model of $\rmc{R}$ given by the RMC functor
  \begin{tikzcd}[cramped,sep=small]
    \rmc{R} \ar{r}{F} & \rmc{S} \ar{r}{\yo} & \PSh{\rmc{S}}
  \end{tikzcd}.
  When $F$ is understood from context, we write $\Synt{\rmc{S}} \in \Mod{\rmc{R}}$.
\end{definition}

A morphism $G \colon (\rmc{S},F) \to (\rmc{S}',F')$ in the coslice (2,1)-category $\Coslice{\rmc{R}}{\RMC}$ induces a morphism $\Synt{G} \colon \Synt{F} \to \Synt{F'}$ of models of $\rmc{R}$.
A special case of the above construction is its application to the identity $\Id \colon \rmc{R} \to \rmc{R}$: the model $\Synt{\rmc{R}} \in \Mod{\rmc{R}}$ is a bi-initial object in $\Mod{\rmc{R}}$, the \emph{initial model} of $\rmc{R}$ \cite[\S5.4.1]{uemura:21}.

\section{Cubical type theories}
\label{sec:cubical}

\subsection{The interval}

Before defining cubical type theory, we first introduce a simple SOGAT specifying the interval alone.
This will let us easily speak about cubical type theories with different interval theories.

\begin{definition}
  The SOGAT $\rmcII$ of an \emph{interval} has one representable sort with two points:
  \[
    \begin{array}{lcl}
      \ccII &:& () \soto \soR \\
    \end{array}
    \qquad\qquad
    \begin{array}{lcl}
      \cce0, \cce1 &:& () \soto \ccII
    \end{array}
  \]
\end{definition}

\begin{definition}
  An \emph{interval theory} is an environment $\Phi \in \rmcII$.
\end{definition}

Per \S\ref{sec:sogats}, a context in $\rmcII$ is simply a list $(\sov{i}_1 : \II, \ldots, \sov{i}_n : \II)$.
An environment $\Phi$ consists of declarations of the form $\sov{r} : \Gamma \to \ccII$ and $\_ : \Gamma \to r_1 \equiv r_2 : \ccII$.
In other words, $\Phi$ specifies a single-sorted algebraic theory extending the theory of two points $\cce0$, $\cce1$.

\begin{example}
  The \emph{cartesian interval theory} $\thCart$ is the trivial environment $1 \defeq () \in \rmcII$.
  The \emph{distributive lattice interval theory} $\thDL$ is the environment beginning with
  \[
    \begin{array}{lcl}
      ({-} \land {-}), ({-} \lor {-}) &:& (\sov{i} : \ccII, \sov{j} : \ccII) \to \ccII \\
      \_ &:& (\sov{i}\;\sov{j}\;\sov{k} : \ccII) \to \sov{i} \land (\sov{j} \lor \sov{k}) \equiv (\sov{i} \land \sov{j}) \lor (\sov{i} \land \sov{k}) : \ccII
    \end{array}
  \]
  and continuing with the other equations of a bounded distributive lattice, as enumerated for example by Buchholtz and Morehouse \cite[Table 1]{buchholtz-morehouse:17}: associativity and commutativity of $\land$ and $\lor$, unit laws $\sov{i} \land \cce1 \equiv \sov{i}$ and $\sov{i} \lor \cce0 \equiv \sov{i}$, and absorption laws $\sov{i} \land (\sov{i} \lor \sov{j}) \equiv \sov{i}$ and $\sov{i} \lor (\sov{i} \land \sov{j}) \equiv \sov{i}$.
\end{example}

\subsection{Cofibrations}

In addition to $\ccTy$, $\ccTm$, and $\ccII$, cubical type theory has a sort of \emph{cofibrations} and, over the sort of cofibrations, a representable sort for \emph{cofibration truth}:
\[
  \begin{array}{lcl}
    \ccCof &:& \soO
  \end{array}
  \qquad\qquad
  \begin{array}{lcl}
    \ccTrue &:& (\sov{P} : \ccCof) \soto \soR
  \end{array}
\]
We write $\rmcCof$ for the SOGAT consisting of these two sorts.
In fact we have $\rmcCof \cong \rmcDTT$, but it will be useful to have distinct notation for this sub-SOGAT of our cubical type theories.

\subsection{Opaque cubical type theory}
\label{sec:opaque}

We define opaque cubical type theory, $\rmcCTT$, as a mutual extension of the SOGATs of Martin-L\"of type theory with $\Sigma$, $\Pi$, and identity types ($\rmcMLTTPi$), an interval ($\rmcII$), and a cofibration classifier ($\rmcCof$).
We roughly follow Uemura's encodings of cubical type theory~\cite[\S4.6.3]{uemura:21}~\cite[Example 5.14]{uemura:23}.
We introduce the declarations of $\rmcCTT$ (beyond those of $\rmcMLTTPi$, $\rmcII$, and $\rmcCof$) in stages over the course of this section (\S\ref{sec:opaque}).

\begin{remark}
  Given that path types serve as equality types in cubical type theories, it may seem strange that we include Martin-L\"of's identity types in $\rmcCTT$, though their coexistence is semantically justified~\cite[\S9.1]{cohen-coquand-huber-mortberg:15}~\cite[\S3.3]{cavallo-harper:19}~\cite[\S2.16]{abcfhl:21}.
  We do so partly in order to reuse Kapulkin and Lumsdaine's tools for comparing type theories \cite{kapulkin-lumsdaine:18}, though we could have redeveloped these with path types.
  The more technical reason is that we want to include (higher) inductive types in $\rmcCTT$.
  The span interpretation (\S\ref{sec:span}) that we use to prove conservativity interprets inductive types as inductive families \cite{dybjer:94}, and we use identity types to define these families.
  This is the one place where identities cannot straightforwardly be replaced with paths.
\end{remark}

\subsubsection{Cofibrations}

In $\rmcCTT$, we add new operators and equations for the sorts of $\rmcCof$.
A cofibration can be thought of as a constraint on interval terms.
Cofibration truth is a \emph{strict proposition} in the sense that any two witnesses to truth of a cofibration are strictly equal:
\[
  \begin{array}{lcl}
    \_ &:& (\sov{P} : \ccCof, \sov{u}\;\sov{v} : \ccTrue(\sov{P})) \soto \sov{u} \equiv \sov{v} : \ccTrue(\sov{P})
  \end{array}
\]

As with $\ccTm$, we will leave the $\ccTrue$ operator implicit.
Cofibrations are closed under finite conjunction ($\ccTop$, $\ccCap$) and disjunction ($\ccBot$, $\ccCup$):
\[
  \begin{array}[t]{lcl}
    \ccTop, \ccBot &:& \ccCof \\
    ({-} \ccCap {-}) &:& (\sov{P}\;\sov{Q} : \ccCof) \soto \ccCof \\
    ({-} \ccCup {-}) &:& (\sov{P}\;\sov{Q} : \ccCof) \soto \ccCof \\
    \_ &:& \top \\
    \_ &:& (\sov{P} : \ccCof, \bot) \soto \sov{P}
  \end{array}
  \qquad
  \begin{array}[t]{lcl}
    \_ &:& (\sov{P}\;\sov{Q} : \ccCof, \sov{P}, \sov{Q}) \soto \sov{P} \ccCap \sov{Q} \\
    \_ &:& (\sov{P}\;\sov{Q} : \ccCof, \sov{P} \ccCap \sov{Q}) \soto \sov{P} \\
    \_ &:& (\sov{P}\;\sov{Q} : \ccCof, \sov{P} \ccCap \sov{Q}) \soto \sov{Q} \\
    \_ &:& (\sov{P}\;\sov{Q} : \ccCof, \sov{P}) \soto \sov{P} \ccCup \sov{Q} \\
    \_ &:& (\sov{P}\;\sov{Q} : \ccCof, \sov{Q}) \soto \sov{P} \ccCup \sov{Q} \\
    \_ &:& (\sov{P}\;\sov{Q}\;\sov{R} : \ccCof, \sov{P} \to \sov{R}, \sov{Q} \to \sov{R}, \sov{P} \ccCup \sov{Q}) \soto \sov{R}
  \end{array}
\]
Eliminators for the nullary and binary disjunction ($\ccAbortTy$, $\ccAbortTm$, $\ccCaseTy$, $\ccCaseTm$) allow us to define types and terms by case analysis.
We abbreviate
\[
  \begin{array}{lcl}
  \Phi_{\ccCup\ccTy} &=& (\sov{P}\;\sov{Q} : \ccCof, \sov{A} : [\sov{P}] \to \ccTy, \sov{B} : [\sov{Q}] \to \ccTy, [\sov{P} \ccCap \sov{Q} \to \sov{A} \equiv \sov{B} : \ccTy]) \\
    \Phi_{\ccCup\ccTm} &=& (\sov{P}\;\sov{Q} : \ccCof, \sov{A} : [\sov{P} \ccCup \sov{Q}] \to \ccTy, \sov{a} : [\sov{P}] \to \sov{A}, \sov{b} : [\sov{Q}] \to \sov{A}, [\sov{P} \ccCap \sov{Q} \to \sov{a} \equiv \sov{b} : \sov{A}])
  \end{array}
\]
and specify
\[
  \def\arraystretch{1.2}
  \begin{array}{lcll}
    \ccAbortTy &:& [\ccBot] \soto \ccTy \\
    \ccAbortTm &:& (\sov{A} : [\ccBot] \to \ccTy, [\ccBot]) \soto \sov{A} \\
    \ccCaseTy &:& (\Phi_{\ccCup\ccTy}, [\sov{P} \ccCup \sov{Q}]) \soto \ccTy \\
    \_ &:& (\Phi_{\ccCup\ccTy}, \sov{P}) \soto \ccCaseTy(\sov{P},\sov{Q},\sov{A},\sov{B}) \equiv \sov{A} : \ccTy \\
    \_ &:& (\Phi_{\ccCup\ccTy}, \sov{Q}) \soto \ccCaseTy(\sov{P},\sov{Q},\sov{A},\sov{B}) \equiv \sov{B} : \ccTy \\
    \ccCaseTm &:& (\Phi_{\ccCup\ccTm}, [\sov{P} \ccCup \sov{Q}]) \soto \sov{A} \\
    \_ &:& (\Phi_{\ccCup\ccTm}, \sov{P}) \soto \ccCaseTm(\sov{P},\sov{Q},\sov{A},\sov{a},\sov{b}) \equiv \sov{a} : \sov{A} \\
    \_ &:& (\Phi_{\ccCup\ccTm}, \sov{Q}) \soto \ccCaseTm(\sov{P},\sov{Q},\sov{A},\sov{a},\sov{b}) \equiv \sov{b} : \sov{A}
  \end{array}
\]
The basic cofibrations are equations on interval terms, which we write with $\ccCofEq$.
The two endpoints $\cce0$ and $\cce1$ are distinct, and we can convert between $\ccCofEq$ and strict equality $\equiv$.
\[
  \def\arraystretch{1.2}
  \begin{array}{lcl}
    {-} \ccCofEq {-} &:& (\sov{i}\;\sov{j} : \ccII) \soto \ccCof \\
    \_ &:& (\cce0 \ccCofEq \cce1) \soto \ccBot
  \end{array}
  \qquad\qquad
  \begin{array}{lcl}
    \_ &:& (\sov{i} : \ccII) \soto \sov{i} \ccCofEq \sov{i} \\
    \_ &:& (\sov{i}\;\sov{j} : \ccII, \sov{i} \ccCofEq \sov{j}) \soto \sov{i} \equiv \sov{j} : \ccII
  \end{array}
\]

\begin{remark}
  We could have included various algebraic laws for cofibrations, such as $\sov{P} \ccCap \sov{Q} \equiv \sov{Q} \ccCap \sov{P}$, or \emph{cofibration extensionality} $(\sov{P}:\ccCof,\sov{Q}:\ccCof,\sov{P} \to \sov{Q},\sov{Q} \to \sov{P}) \to \sov{P} \equiv \sov{Q} : \ccCof$.
  Our proofs go through for such variations without much change.
\end{remark}

\subsubsection{Filling}

Cofibrations are used to specify the \emph{filling operator}.
We first introduce the abbreviation $\Phi_{\ccFill}$ for the environment
\[
  \begin{array}{l}
    (\sov{A} : \ccII \to \ccTy, \sov{P} : \ccCof, \sov{a} : ([\sov{P}], \sov{i} : \ccII) \to \sov{A}(\sov{i}), \sov{j} : \ccII, \sov{a}_0 : \sov{A}(\sov{j}), [\sov{P} \to \sov{a}(\sov{j}) \equiv \sov{a}_0 : \sov{A}(\sov{j})]) \rlap{.}
  \end{array}
\]
This environment specifies a line ($\ccII$-indexed family) of types $\sov{A}$ and a ``partial'' line of terms $\sov{a}$ over it, defined whenever some cofibration $\sov{P}$ is true, together with a fully-defined term $\sov{a}_0$ at some index $\sov{A}(\sov{j})$ that coincides with $\sov{a}(\sov{j})$ when $\sov{P}$ holds.
Given this input, the filling operator outputs a line $(\sov{k} : \ccII) \to \sov{A}(\sov{k})$ that ``extends'' both $\sov{a}$ and $\sov{a}_0$ in the following sense.
\[
  \begin{array}{lcl}
    \ccFill &:& (\Phi_{\ccFill}, \sov{k} : \ccII) \soto \sov{A}(\sov{k}) \\
    \_ &:& (\Phi_{\ccFill}, \sov{k} : \ccII, \sov{P}) \soto \ccFill(\sov{A},\sov{P},\sov{a},\sov{j},\sov{a}_0,\sov{k}) \equiv \sov{a}(\sov{k}) : \sov{A}(\sov{k}) \\
    \_ &:& (\Phi_{\ccFill}) \soto \ccFill(\sov{A},\sov{P},\sov{a},\sov{j},\sov{a}_0,\sov{j}) \equiv \sov{a}_0 : \sov{A}(\sov{j})
  \end{array}
\]

The special case where $\sov{P} = \ccBot$ is called \emph{coercion} by Angiuli et al.~\cite[\S2.7]{abcfhl:21} and converts a term at some index $\sov{a}_0 : \sov{A}(\sov{j})$ to a term at any other index $\sov{A}(\sov{k})$.

\begin{notation}
  Over the environment $(\sov{A} : (\sov{i} : \ccII) \to \ccTy, \sov{j} : \ccII, \sov{a}_0 : \sov{A}(\sov{j}), \sov{k} : \ccII)$, write $\ccCoe(\sov{A},\sov{j},\sov{a}_0,\sov{k}) \defeq \ccFill(\sov{A},\ccBot,\soabs{\sov{i}}\ccAbortTm(\sov{A}(\sov{i})), \sov{j}, \sov{a}_0, \sov{k}) : \sov{A}(\sov{k})$.
\end{notation}

\begin{notation}
  We write $\ccFill^{j \to k}(A, [P_1 \mapsto a_1, \ldots, P_n \mapsto a_n], a_0)$ for
  $\ccFill(A,P,a,j,a_0,k)$ where $P = P_1 \ccCup \cdots \ccCup P_n$ (with some choice of parentheses) and $a$ is defined from $a_1,\ldots,a_n$ by cases using $\ccCaseTm$.
  We write $\ccCoe^{j \to k}(A,a_0)$ for $\ccCoe(A,j,a_0,k)$.
\end{notation}

\begin{remark}
  This definition of $\ccFill$ is a suitable base for strict cubical type theory over arbitrary interval theories.
  In the presence of certain interval structure, it can be reduced to special cases.
  For theories with connections, $\ccFill^{\cce0 \to \cce1}$ and $\ccFill^{\cce1 \to \cce0}$ suffice; see Cavallo, M\"ortberg, and Swan~\cite[Theorem 14 with Lemma 8]{cavallo-mortberg-swan:20}.
  With two connections and a reversal, this can be further reduced to $\ccFill^{\cce0 \to \cce1}$, as in Cohen et al.'s type theory~\cite{cohen-coquand-huber-mortberg:15}; see Angiuli et al.~\cite[\S3.4]{abcfhl:21}.
  We refer to Cavallo, M\"ortberg, and Swan~\cite{cavallo-mortberg-swan:20} for more detailed comparisons.
\end{remark}

\subsubsection{Paths}

A path is an $\ccII$-indexed term taking two fixed values at the endpoints $\cce0,\cce1 : \ccII$.
Path types internalize paths:
\[
  \begin{array}{lcl}
    \ccPath &:& (\sov{A} : (\sov{i} : \ccII) \to \ccTy, \sov{a}_0 : \sov{A}(\cce0), \sov{a_1} : \sov{A}(\cce1)) \soto \ccTy \\
    \ccPathLam &:& ([\sov{A} : (\sov{i} : \ccII) \to \ccTy], \sov{a} : (\sov{i} : \ccII) \to \sov{A}(i)) \soto \ccPath(\sov{A},\sov{a}(\cce0),\sov{a}(\cce1)) \\
    {-} \ccPathApp {-} &:& ([\sov{A} : (\sov{i} : \ccII) \to \ccTy, \sov{a}_0 : \sov{A}(\cce0), \sov{a_1} : \sov{A}(\cce1)], \sov{p} : \ccPath(\sov{A},\sov{a}_0,\sov{a_1}), \sov{i} : \ccII) \soto \sov{A}(\sov{i})
  \end{array}
\]
The equations for path types state that $\ccPathLam(\sov{a}) \ccPathApp \sov{i} \equiv \sov{a}(\sov{i})$ and $\sov{p} \equiv \ccPathLam(\soabs{\sov{i}}\sov{p} \ccPathApp \sov{i})$, as for function types, as well as that $\sov{p} \ccPathApp \cce0 \equiv \sov{a}_0$ and $\sov{p} \ccPathApp \cce1 \equiv \sov{a}_1$ for $\sov{p} : \ccPath(\sov{A},\sov{a}(\cce0),\sov{a}(\cce1))$.
See Uemura \cite[\S4.6.3, Type constructors]{uemura:21} for a fully formal presentation.

\begin{notation}
  We write $\ccPathLam* \sov{i}.a$ as shorthand for $\ccPathLam \soabs{\sov{i}} a$.
  We abbreviate ``non-dependent'' path types $\ccPath(\soabs{\_}A,a_0,a_1)$, where the line of types is constant, as $a_0 \sim^A a_1$ or simply $a_0 \sim a_1$.
\end{notation}

\subsubsection{Glue types}

In cubical type theories, univalence is not an axiom but is instead derived from a type former that can construct $\ccII$-indexed types from equivalences.
Universes closed under this type former can be shown to be univalent.
Following Cohen et al.~\cite{cohen-coquand-huber-mortberg:15}, Angiuli et al.\ implement univalence using so-called glue types \cite[\S2.11]{abcfhl:21}; Angiuli, Favonia, and Harper's $\mathsf{V}$-types are an alternative solution \cite[\S5.6]{angiuli-favonia-harper:18}.

First, we define equivalences \cite[Definition 4.4.1]{hott-book} using path types.
Over $\sov{A} : \ccTy$, define $\ccIsContr(\sov{A}) \defeq \ccSigma \sov{a}_0{:}\sov{A}.\ \ccPi \sov{a}_1{:}\sov{A}.\ \sov{a}_0 \sim^{\sov{A}} \sov{a}_1 : \ccTy$.
Over $([\sov{A}\;\sov{B} : \ccTy], \sov{f} : \sov{A} \to \sov{B})$, define $\ccIsEquiv(\sov{f}) \defeq \Pi \sov{b}:\sov{B}. \ccIsContr(\ccSigma \sov{a}:\sov{A}.\ \sov{f}(\sov{a}) \sim \sov{b}) : \ccTy$.
Finally, over $(\sov{A}\;\sov{B} : \ccTy)$, define the type of \emph{equivalences} $(\sov{A} \simeq \sov{B}) \defeq \Sigma \sov{f}:\sov{A} \to \sov{B}.\ \ccIsEquiv(\sov{f})$.
The $\ccGlue$ type former takes a type $\sov{A}$, a cofibration $\sov{P}$, and a partial type $\sov{T}$ and equivalence $\sov{e} : \sov{T} \simeq \sov{A}$ defined when $\sov{P}$ holds.
Its output is a total type that reduces to $\sov{T}$ when $\sov{P}$ holds.
\[
  \begin{array}{lcl}
    \ccGlue &:& (\sov{A} : \ccTy, \sov{P} : \ccCof, \sov{T} : [\sov{P}] \to \ccTy, \sov{e} : [\sov{P}] \to \sov{T} \simeq \sov{A}) \soto \ccTy \\
    \_ &:& (\sov{A} : \ccTy, \sov{P} : \ccCof, \sov{T} : [\sov{P}] \to \ccTy, \sov{e} : [\sov{P}] \to \sov{T} \simeq \sov{A}, \sov{P}) \soto \ccGlue(\sov{A},\sov{P},\sov{T},\sov{e}) \equiv \sov{T} : \ccTy
  \end{array}
\]
We now abbreviate $\Phi_{\ccGlue} = (\sov{A} : \ccTy, \sov{P} : \ccCof, \sov{T} : [\sov{P}] \to \ccTy, \sov{e} : [\sov{P}] \to \sov{T} \simeq \sov{A})$.
The $\ccGlue$ type has an introduction form $\ccEnglue$ and an elimination form $\ccUnglue$.
Each reduces when $\sov{P}$ holds, and we have computation and uniqueness equations.
\[
  \begin{array}{lcl}
    \ccEnglue &:& ([\Phi_{\ccGlue}], \sov{a} : \sov{A}, \sov{t} : [\sov{P}] \to \sov{T}, [\sov{P} \to \sov{e}.1(\sov{t}) \equiv \sov{a} : \sov{A}]) \soto \ccGlue(\sov{A},\sov{P},\sov{T},\sov{e}) \\
    \_ &:& ([\Phi_{\ccGlue}], \sov{a} : \sov{A}, \sov{t} : [\sov{P}] \to \sov{T}, [\sov{P} \to \sov{e}.1(\sov{t}) \equiv \sov{a} : \sov{A}], \sov{P}) \soto \ccEnglue(\sov{a},\sov{t}) \equiv \sov{t} : \sov{T} \\
    \ccUnglue &:& ([\Phi_{\ccGlue}], \sov{g} : \ccGlue(\sov{A},\sov{P},\sov{T},\sov{e})) \soto \sov{A} \\
    \_ &:& ([\Phi_{\ccGlue}], \sov{g} : \ccGlue(\sov{A},\sov{P},\sov{T},\sov{e}), \sov{P}) \soto \ccUnglue(\sov{g}) \equiv \sov{e}.1(\sov{g}) : \sov{A} \\
    \_ &:& ([\Phi_{\ccGlue}], \sov{a} : \sov{A}, \sov{t} : [\sov{P}] \to \sov{T}, [\sov{P} \to \sov{e}.1(\sov{t}) \equiv \sov{a} : \sov{A}]) \soto \ccUnglue(\ccEnglue(\sov{a},\sov{t})) \equiv \sov{a} : \sov{A} \\
    \_ &:& ([\Phi_{\ccGlue}], \sov{g} : \ccGlue(\sov{A},\sov{P},\sov{T},\sov{e})) \soto \sov{g} \equiv \ccEnglue(\ccUnglue(\sov{g}),\sov{g}) : \ccGlue(\sov{A},\sov{P},\sov{T},\sov{e}) \\
  \end{array}
\]
The eliminator $\ccUnglue$ can be shown to be an equivalence $\ccGlue(\sov{A},\sov{P},\sov{T},\sov{e}) \simeq \sov{A}$ that reduces to $\sov{e}$ when $\sov{P}$ holds.
Univalence is derived using an instance where $\sov{P}$ is $(\sov{i} \ccCofEq \cce0 \ccCup \sov{i} \ccCofEq \cce1)$ for some $\sov{i} : \ccII$; see Cohen et al.~\cite[\S7.2]{cohen-coquand-huber-mortberg:15} or Angiuli et al.~\cite[\S2.12]{abcfhl:21}.

\subsubsection{Universe}

We include one universe: a type $\ccU$ whose elements are regarded as types via a decoding function $\ccEl$.
\[
  \begin{array}{lcl}
    \ccU &:& \ccTy
  \end{array}
  \qquad\qquad
  \begin{array}{lcl}
    \ccEl &:& (\sov{A} : \ccU) \soto \ccTy
  \end{array}
\]
We often leave the coercion $\ccEl$ from $\ccU$ to types implicit.
For a universe to be useful, it should be closed under type formers such as $\ccSigma$ and $\ccPi$; for it to be univalent, it should be closed under $\ccGlue$.
We refer to Uemura \cite[Example 4.6.11]{uemura:21} for an example formulation of these closure conditions, but we omit cases for type formers in the universe in our proofs: handling them always amounts to repeating the construction used for the type formers outside of the universe.

\subsubsection{Higher inductive types}

We include \emph{suspension} types \cite[\S6.5]{hott-book} as a representative example of an HIT.
See \cite{coquand-huber-mortberg:18,cavallo-harper:19} for general descriptions of HITs in cubical type theory.
We specify the formation and introduction forms as follows:
\[
  \begin{array}[t]{lcl}
    \ccSusp &:& (\sov{A} : \ccTy) \soto \ccTy
  \end{array}
  \qquad
  \begin{array}[t]{lcl}
    \ccNorth, \ccSouth &:& [\sov{A} : \ccTy] \soto \ccSusp(\sov{A}) \\
    \ccMerid &:& ([\sov{A} : \ccTy], \sov{a} : \sov{A}) \soto \ccNorth \sim^{\ccSusp(\sov{A})} \ccSouth
  \end{array}
\]
For elimination, we fix the environment
\[
  \begin{array}{lcl}
    \Phi_{\ccSuspElim} &=& ([\sov{A} : \ccTy], \sov{C} : (\sov{t} : \ccSusp(\sov{A})) \to \ccTy, \\
                       && \phantom{(} \sov{n} : \sov{C}(\ccNorth), \sov{s} : \sov{C}(\ccSouth), \sov{m} : (\sov{a} : \sov{A}) \to \ccPath(\soabs{\sov{i}}\sov{C}(\ccMerid(\sov{a}) \ccPathApp \sov{i}), \sov{n}, \sov{s}))
  \end{array}
\]
and specify
\[
  \begin{array}{lcl}
    \ccSuspElim &:& (\Phi_{\ccSuspElim}, \sov{t} : \ccSusp(\sov{A})) \soto \sov{C}(\sov{t}) \\
    \_ &:& (\Phi_{\ccSuspElim}) \soto \ccSuspElim(\sov{C},\sov{n},\sov{s},\sov{m},\ccNorth) \equiv \sov{n} : \sov{C}(\ccNorth) \\
    \_ &:& (\Phi_{\ccSuspElim}) \soto \ccSuspElim(\sov{C},\sov{n},\sov{s},\sov{m},\ccSouth) \equiv \sov{s} : \sov{C}(\ccSouth) \\
    \ccSuspElimMerid &:& (\Phi_{\ccSuspElim}, \sov{a} : \sov{A}) \soto \ccPathLam*\sov{i}.\ccSuspElim(\sov{C},\sov{n},\sov{s},\sov{m},\ccMerid(\sov{a}) \ccPathApp \sov{i}) \sim  \sov{m}
  \end{array}
\]
This is an ``opaque'' suspension type in that $\ccSuspElimMerid$ constructor is a path rather than a strict equality.
This is how HITs are usually formulated in Book HoTT \cite[\S6.2]{hott-book}, strict computation rules being characteristic of cubical type theory.

\subsection{Strict cubical type theory}

\emph{Strict} cubical type theories---\ie, cubical type theories as they are usually defined---are designed to satisfy strict canonicity \cite{huber:18,angiuli-favonia-harper:18}, the property that every closed term of type $\mathbb{N}$ is strictly equal to a numeral.
This requires two adjustments to our opaque cubical type theory, which we sketch here.
A full description of the specific strict theory we model in \S\ref{sec:spaces} can be found in Angiuli et al.~\cite{abcfhl:21}.
We write $\rmcCTTs$ for the extension of the SOGAT $\rmcCTT$ with the symbols and equations indicated below, following Angiuli et al.'s specification.

First, we add equations for each concrete type former for evaluating applications of the filling operator at that type.
For $\Sigma$ types, for example, we have an equation reducing $\ccFill(\soabs{\sov{i}}\ccSigma\sov{a}{:}\sov{A}(\sov{i}).\sov{B}(\sov{i},\sov{a}),\sov{P},\sov{s},\sov{j},\sov{s}_0,\sov{k})$ to a pair of two calls to the filling operator, one over $\sov{A}$ and one over an instance of $\sov{B}$.
For higher inductive types such as $\ccSusp$, some applications of the filling operator are treated as values (\ie, not reduced), and equations are instead introduced for reducing the eliminator at these values \cite[\S2.15]{abcfhl:21}.

Second, we strictify the path $\ccSuspElimMerid$, replacing it with a strict equation or, to express strict cubical type theory as an extension of opaque cubical type theory, introducing the strict equation and equating $\ccSuspElimMerid$ with the reflexive path.

\section{The twist interpretation}
\label{sec:twist}

To prove conservativity of opaque cubical type theories with reversals over the corresponding theories without reversals, we first construct interpretations from the former to the latter.
In \S\S\ref{sec:span}--\ref{sec:conservativity} we show that the existence of these interpretations abstractly implies conservativity.
As sketched in \S\ref{sec:contributions}, we exploit \emph{twist constructions}: the fact that the ``square'' environment $\ccII \times \ccII = (\sov{i}_0 : \ccII, \sov{i}_1 : \ccII) \in \rmcCTT$ is an interval object with a reversal and inherits certain algebraic structure from $\ccII$.
Thus, we call our translation the \emph{twist interpretation}.

\subsection{Extension by a reversal}
\label{sec:extension-by-reversal}

We have an interpretation $\flipII$ of $\rmcII$ in itself by taking $\flipII(\ccII) \defeq \ccII$, $\flipII(\cce0) \defeq \cce1$, and $\flipII(\cce1) \defeq \cce0$.
By the (2, 1)-categorical universal property of $\rmcCl{\rmcII}$ \cite[Theorem 4.8.18]{uemura:21}, this determines (up to isomorphism) an RMC functor $\flipII \colon \rmcII \to \rmcII$ with an isomorphism $\theta \colon \flipII \circ \flipII \cong \Id$ satisfying $\theta \circ \flipII = \flipII \circ \theta$ and $\theta_{\ccII} = \id$.

\begin{definition}
  A \emph{self-dual interval theory} $(\Phi,\phi)$ is an interval theory $\Phi$ equipped with an isomorphism $\phi \colon \flipII(\Phi) \cong \Phi$ such that $\flipII(\phi) \circ \phi = \theta_\Phi : \flipII(\flipII(\Phi)) \cong \Phi$.
\end{definition}

In a self-dual interval theory $(\Phi,\phi)$, the value of $\phi$ at an operator $\sov{r} : \II^n \to \II$ in $\Phi$ is an expression $\phi(\sov{r}) : \II^n \to \II$ over $\Phi$: the dual of $\sov{r}$.

\begin{example}
  \label{example-dual-theories}
  The cartesian theory $\thCart$ is self-dual with the trivial isomorphism $1 \cong 1$.
  The theory $\thDL$ is self-dual with $\phi$ defined by $\phi(- \land -)(\sov{i},\sov{j}) = \sov{i} \lor \sov{j}$ and vice versa.
\end{example}

\begin{definition}
  Given a self-dual interval theory $(\Phi,\phi)$, its \emph{extension by a reversal} $\thRev{\Phi}{\phi} \in \rmcII$ is the extension of $\Phi$ with
  \begin{enumerate}[(a)]
  \item an operator $\lnot : \ccII \to \ccII$,
  \item equations $\lnot \cce0 \equiv \cce1 : \ccII$, $\lnot \cce1 \equiv \cce0 : \ccII$, and $(\sov{i} : \ccII) \to \lnot (\lnot (\sov{i})) \equiv \sov{i} : \ccII$, and
  \item for each $\sov{r} : (\sov{i}_1 : \II, \ldots, \sov{i}_n : \II) \to \II$ in $\Phi$, an equation
    \[
      (\sov{i}_1 : \II, \ldots, \sov{i}_n : \II) \to \lnot(\sov{r}(\sov{i}_1,\ldots,\sov{i}_n)) \equiv \phi(\sov{r})(\lnot(\sov{i}_1),\ldots,\lnot(\sov{i}_n)) : \II \rlap{.}
    \]
  \end{enumerate}
\end{definition}

\begin{example}
  The interval theory $\thRev{\thCart}{\phi}$ for $\phi \colon 1 \cong 1$ consists simply of the operator $\lnot : \II \to \II$ and equations $\lnot \cce0 \equiv \cce1 : \ccII$, $\lnot \cce1 \equiv \cce0 : \ccII$, and $(\sov{i} : \ccII) \to \lnot (\lnot (\sov{i})) \equiv \sov{i} : \ccII$.
  The interval theory $\thRev{\thDL}{\phi}$, for the isomorphism $\phi \colon \thDL \cong \thDL$ from \Cref{example-dual-theories}, is the algebraic theory of a De Morgan algebra bounded by $\cce0$ and $\cce1$.
\end{example}

\begin{definition}[Twist interpretation of the interval]
  \label{interval-twist}
  For a self-dual interval theory $(\Phi,\phi)$, we define a representable map functor $\twistII \colon \rmcII[\thRev{\Phi}{\phi}] \to \rmcII[\Phi]$ by the following interpretation:
  \begin{enumerate}
  \item On sorts, we set $\twistII\II \defeq \II \times \II$.
  \item We interpret $\cce0$ as $(\cce0,\cce1)$ and $\cce1$ as $(\cce1,\cce0)$.
  \item We interpret each $\sov{r} : (\sov{i}_1 : \II, \ldots, \sov{i}_n : \II) \to \II$ in $\Phi$ by
    \[
      \twistII\sov{r}((\sov{i}_{10},\sov{i}_{11}), \ldots, (\sov{i}_{n0},\sov{i}_{n1})) \defeq (\sov{r}(\sov{i}_{10},\ldots,\sov{i}_{n0}), \phi(\sov{r})(\sov{i}_{11},\ldots,\sov{i}_{n1})) \rlap{.}
    \]
  \item We interpret $\lnot$ by $\twistII\lnot((\sov{i}_0,\sov{i}_1)) \defeq (\sov{i}_1,\sov{i}_0)$.
  \end{enumerate}
\end{definition}

\subsection{Interpreting cubical type theory}

Cubical type theory being an extension of the theory of an interval, any environment $\Phi$ over $\rmcII$ can also be regarded as an environment $\iota\Phi$ over $\rmcCTT$, from which we can produce a new SOGAT $\rmcCTT[\iota\Phi]$: cubical type theory with the interval theory $\Phi$.

We now extend $\twistII \colon \rmcII[\thRev{\Phi}{\phi}] \to \rmcII[\Phi]$ for a self-dual interval theory $(\Phi,\phi)$ to an interpretation $\twistCTT \colon \rmcCTT[\iota\thRev{\Phi}{\phi}] \to \rmcCTT[\iota\Phi]$.
The specification of this interpretation occupies the remainder of this section; we summarize in \Cref{twist-theorem}.

\begin{component}[$\twistCTT$, sorts]
  We set $\twistCTT\ccII \defeq \ccII \times \ccII$ and interpret all other sorts by themselves: $\twistCTT\ccTy \defeq \ccTy$, $(\twistCTT\ccTm)(\sov{A}) \defeq \ccTm(\sov{A})$, $\twistCTT\ccCof \defeq \ccCof$, $(\twistCTT\ccTrue)(\sov{P}) \defeq \ccTrue(\sov{P})$.
\end{component}

\begin{notation}
  For infix operators, we use a subscript to denote interpretation, for example writing $\ccCofEq_\twistCTT$ instead of $\twistCTT(- \ccCofEq -)$.
\end{notation}

\begin{component}[$\twistCTT$, interval theory]
  We interpret the operations of $\thRev{\Phi}{\phi}$ in $\rmcCTT[\iota\Phi]$ as in \Cref{interval-twist}.
\end{component}

\begin{component}[$\twistCTT$, cofibration theory]
  We interpret the cofibration-forming operations as follows.
  \[
    \def\arraystretch{1.6}
    \begin{array}{c}
      (\sov{i}_0,\sov{i}_1) \ccCofEq_\twistCTT (\sov{j}_0,\sov{j}_1) \defeq (\sov{i}_0 \ccCofEq \sov{j}_0) \ccCap (\sov{i}_1 \ccCofEq \sov{j}_1) \\
      \twistCTT\ccTop \defeq \ccTop \qquad\qquad
    \sov{P} \ccCap_\twistCTT \sov{Q} \defeq \sov{P} \ccCap \sov{Q} \qquad\qquad
      \twistCTT\ccBot \defeq \ccBot \qquad\qquad
      \sov{P} \ccCup_\twistCTT \sov{Q}  \defeq \sov{P} \ccCup \sov{Q}
    \end{array}
  \]
  These definitions validate the associated axioms for the $\ccTrue$ judgment.
  We interpret the $\ccAbortTy$, $\ccAbortTm$, $\ccCaseTy$, and $\ccCaseTm$ eliminators as themselves.
\end{component}

\begin{component}[$\twistCTT$, filling]
  \label{twist-filling}
  We interpret filling by iterated filling, first in one component of the interval variable and then in the other, defining $\twistCTT\ccFill(\sov{A},\sov{P},\sov{a},(\sov{j}_0,\sov{j}_1),\sov{a}_0,(\sov{k}_0,\sov{k}_1))$ to be
  \[
    \ccFill^{\sov{j}_1 \to \sov{k}_1}(\soabs{\sov{i}_1}\sov{A}(\sov{k}_0,\sov{i}_1),[\sov{P} \mapsto \soabs{\sov{i}_1}\sov{a}(\sov{k}_0,\sov{i}_1)], \ccFill^{\sov{j}_0\to\sov{k}_0}(\soabs{\sov{i}_0}\sov{A}(\sov{i}_0,\sov{j}_1),[\sov{P} \mapsto \soabs{\sov{i}_0}\sov{a}(\sov{i}_0,\sov{j}_1)],\sov{a}_0)) \rlap{.}
  \]
\end{component}

We interpret type formers that do not involve the interval, namely $\Sigma$ and $\Pi$ types, identity types, and $\ccU$ and $\ccEl$, as themselves.
This leaves path types, glue types, and suspensions.
We interpret the path type as a type of squares with fixed values at the coordinates $\twistCTT\cce0 = (\cce0,\cce1)$ and $\twistCTT\cce1 = (\cce1,\cce0)$, encoded as an iterated path type consisting of the two unfixed points at $(\cce0,\cce0)$ and $(\cce1,\cce1)$, four 1-dimensional paths forming a boundary, and a 2-dimensional path relating them.

\begin{component}[$\twistCTT$, path types]
  \label{twist-path}
  We define $\twistCTT\ccPath(\sov{A},\sov{a}_{01},\sov{a}_{10})$ to be the iterated path type
  \[
    \begin{array}{l}
      \ccSigma \sov{a}_{00}{:}\sov{A}(\cce0,\cce0).\ \ccSigma\sov{a}_{11}{:}\sov{A}(\cce1,\cce1).\ \\
      \ccSigma\sov{p}_{\bullet0}{:}\ccPath(\soabs{\sov{i}_0}\sov{A}(\sov{i}_0,\cce0),\sov{a}_{00},\sov{a}_{10}).\ \ccSigma\sov{p}_{\bullet1}{:}\ccPath(\soabs{\sov{i}_0}\sov{A}(\sov{i}_0,\cce1),\sov{a}_{01},\sov{a}_{11}).\ \\
      \ccSigma\sov{p}_{0\bullet}{:}\ccPath(\soabs{\sov{i}_1}\sov{A}(\cce0,\sov{i}_1),\sov{a}_{00},\sov{a}_{01}).\ \ccSigma\sov{p}_{1\bullet}{:}\ccPath(\soabs{\sov{i}_1}\sov{A}(\cce1,\sov{i}_1), \sov{a}_{10}, \sov{a}_{11}).\ \\
      \ccPath(\soabs{\sov{i}_0}\ccPath(\soabs{\sov{i}_1}\sov{A}(\sov{i}_0,\sov{i}_1),\sov{p}_{\bullet0} \ccPathApp \sov{i}_0,\sov{p}_{\bullet1} \ccPathApp \sov{i}_0),\sov{p}_{0\bullet},\sov{p}_{1\bullet})
    \end{array}
    \quad
    \begin{tikzcd}
      a_{00} \ar[-]{d}[left]{p_{0\bullet}} \ar[-]{r}{p_{\bullet0}} &[-1em] a_{10} \ar[-]{d}{p_{1\bullet}} \\
      a_{01} \ar[-]{r}[below]{p_{\bullet1}} & a_{11}
    \end{tikzcd}
  \]
  and set
  $
  \twistCTT\ccPathLam(\sov{a}) \defeq (\_,\_,\_,\_,\_,\ccPathLam*\sov{i}_0.\ccPathLam*\sov{i}_1.\sov{a}(\sov{i}_0,\sov{i}_1))
  $
  and
  $\sov{t} \ccPathApp_\twistCTT (\sov{i}_0,\sov{i}_1) \defeq \sov{t}.6 \ccPathApp \sov{i}_0 \ccPathApp \sov{i}_1$, where the first five components in $\twistCTT\ccPathLam$ are determined by the final one.
  
  We write $\twistCTT\ccPathLam* \sov{i}_0,\sov{i}_1.a$ as shorthand for $\twistCTT\ccPathLam (\soabs{\sov{i}_0,\sov{i}_1}a)$.
\end{component}

\begin{remark}
  This iterated path type could be naturally expressed as an \emph{extension type}.
  Introduced by Riehl and Shulman for simplicial type theory \cite[\S2.2]{riehl-shulman:17} and discussed by Angiuli \cite[\S3.5]{angiuli:19} in the context of cubical type theory, these are types of $n$-cubes with fixed values on some cofibration.
  In a theory with these types, $\twistCTT\ccPath$ could be defined as an extension type over the cofibration in two variables $\sov{i}_0 : \ccII, \sov{i}_1 : \ccII \vdash (\sov{i}_0 \ccCofEq \cce0 \ccCap \sov{i}_1 \ccCofEq \cce1) \ccCup (\sov{i}_0 \ccCofEq \cce1 \ccCap \sov{i}_1 \ccCofEq \cce0)$.
\end{remark}

To interpret glue and suspension types, we need to convert between inhabitants of $\ccPath(\sov{C}, \sov{c}_0, \sov{c}_1)$ and inhabitants of $\twistCTT\ccPath(\soabs{\sov{i}_0,\sov{i}_1}\sov{C}(\sov{i}_0), \sov{c}_0, \sov{c}_1)$.
First, the easy direction:

\begin{notation}
  Over the environment $([\sov{C} : \ccII \to \ccTy, \sov{c}_0 : \sov{C}(\cce0), \sov{c}_1 : \sov{C}(\cce1)], \sov{p} : \ccPath(\sov{C},\sov{c}_0,\sov{c}_1))$, we define
  $\ccThicken(\sov{p}) \defeq \twistCTT\ccPathLam*\sov{i}_0,\sov{i}_1. \sov{p} \ccPathApp \sov{i}_0 : \twistCTT\ccPath(\soabs{\sov{i}_0,\sov{i}_1}\sov{C}(\sov{i}_0), \sov{c}_0, \sov{c}_1)$.
\end{notation}

For the inverse, we extract the ``anti-diagonal'' of a square by inverting it along one axis---a standard construction using the filling operation---and then extracting the diagonal.

\begin{definition}[Path inversion]
  Over the environment $([\sov{C} : \ccTy, \sov{c}_0\;\sov{c}_1 : \sov{C}], \sov{p} : \sov{c}_0 \sim^{\sov{C}} \sov{c}_1)$, we define
  $\ccSym(\sov{p}) \defeq \ccPathLam*\sov{i}. \ccFill^{\cce1 \to \cce0}(\soabs{\sov{\_}}\sov{C},[\sov{i} \ccCofEq \cce0 \mapsto \soabs{\_}\sov{c}_1, \sov{i} \ccCofEq \cce1 \mapsto \soabs{\sov{j}}\sov{p} \ccPathApp \sov{j}],\sov{c}_1) : \sov{c}_1 \sim^{\sov{C}} \sov{c}_0$.
\end{definition}

\begin{definition}
  \label{antidiagonal}
  Over $([\sov{C} : \ccII \to \ccTy, \sov{c}_0 : \sov{C}(\cce0), \sov{c}_1 : \sov{C}(\cce1)], \sov{q} : \twistCTT\ccPath(\soabs{\sov{i}_0,\sov{i}_1}\sov{C}(\sov{i}_0), \sov{c}_0, \sov{c}_1))$, we define
  $\ccAnti(\sov{q}) \defeq \ccPathLam*\sov{i}. \ccSym(\soabs{\sov{j}}\sov{q} \ccPathApp_\twistCTT (\sov{i},\sov{j})) \ccPathApp \sov{i} : \ccPath(\sov{C},\sov{c}_0,\sov{c}_1)$.
\end{definition}

To show that these constitute an equivalence, we use the contractibility of dependent singleton types:

\begin{proposition}
  \label{singleton-contractibility}
  Over the environment $(\sov{C} : \ccII \to \ccTy, \sov{c}_0 : \sov{C}(\cce0))$, we have a term of type $\ccIsContr(\ccSigma \sov{c}_1{:}\sov{C}(\cce1). \ccPath(\sov{C},\sov{c}_0,\sov{c}_1))$.
\end{proposition}
\begin{proof}[{{{Proof (cf.\ \cite[\S3.2]{angiuli:19})}}}]
  For the center of contraction, take the pair $\sov{s}_0 \defeq \pair{\_}{\ccPathLam* \sov{i}. \ccCoe^{\cce0\to\sov{i}}(\sov{C},\sov{c}_0)}$ (whose first component is determined by its second).
  Given a singleton $\sov{s}$, we have a path
  $\ccPathLam*\sov{j}. \pair{\_}{\ccPathLam* \sov{i}. \ccFill^{\cce0\to\sov{i}}(\sov{C},[\sov{j} \ccCofEq \cce0 \mapsto \soabs{\sov{k}} \sov{s}_0 \ccPathApp \sov{k}, \sov{j} \ccCofEq \cce1 \mapsto \soabs{\sov{k}}\sov{s} \ccPathApp \sov{k}],\sov{c}_0)}$
  from $\sov{s}_0$ to $\sov{s}$.
\end{proof}

\begin{lemma}
  \label{thicken-is-equiv}
  Over the environment $(\sov{C} : \ccTy, \sov{c}_0 : \sov{C}, \sov{c}_1 : \sov{C})$, the function $\lambda \sov{p}. \ccThicken(\sov{p}) : \ccPath(\sov{C},\sov{c}_0,\sov{c}_1) \to \twistCTT\ccPath(\soabs{\sov{i}_0,\sov{i}_1}\sov{C}(\sov{i}_0),\sov{c}_0,\sov{c}_1)$ is an equivalence.
\end{lemma}
\begin{proof}
  Our proof of \Cref{singleton-contractibility} uses only constructs on which we have already defined $\twistCTT$.
  Thus we can mechanically derive from it a term of type $\twistCTT\ccIsContr(\ccSigma \sov{d}_1{:}\sov{D}(\cce1,\cce0). \twistCTT\ccPath(\sov{D},\sov{d}_0,\sov{d}_1))$ over $(\sov{D} : \ccII \times \ccII \to \ccTy, \sov{d}_0 : \sov{D}(\cce0,\cce1))$.
  Using \Cref{antidiagonal}, we can go from $\twistCTT\ccIsContr$ to $\ccIsContr$.
  Taking $\sov{D}(\sov{i}_0,\sov{i}_1) \defeq \sov{C}(\sov{i}_0)$ and $\sov{d}_0 \defeq \sov{c}_0$ gives $\ccIsContr(\ccSigma \sov{c}_1{:}\sov{C}(\cce1). \twistCTT\ccPath(\soabs{\sov{i}_0,\sov{i}_1}\sov{C}(\sov{i}_0),\sov{c}_0,\sov{c}_1))$.
  Thus $\lambda \sov{s}. (\sov{s}.1,\ccThicken(\sov{s}.2)) : \ccSigma \sov{c}_1{:}\sov{C}(\cce1). \ccPath(\sov{C},\sov{c}_0,\sov{c}_1) \to \ccSigma \sov{c}_1{:}\sov{C}(\cce1). \twistCTT\ccPath(\soabs{\sov{i}_0,\sov{i}_1}\sov{C}(\sov{i}_0),\sov{c}_0,\sov{c}_1)$ is a map between contractible types and therefore an equivalence.
  It follows \cite[Theorem 4.7.7]{hott-book} that it is also a fiberwise equivalence.
\end{proof}

We can use \Cref{thicken-is-equiv} inside the definition of equivalence to construct $\twistCTT\ccGlue$.

\begin{component}[$\twistCTT$, glue]
  Over $(\sov{A} : \ccTy, \sov{P} : \ccCof, \sov{T} : [\sov{P}] \to \ccTy, \sov{e} : [\sov{P}] \to \sov{T} \simeq_\twistCTT \sov{A})$, define $\twistCTT\ccGlue(\sov{A},\sov{P},\sov{T},\sov{e}) \defeq \ccGlue(\sov{A},\sov{P},\sov{T},\widehat{\mathsf{e}})$ where $\widehat{\mathsf{e}}$ is derived from $\sov{e}$ by using \Cref{thicken-is-equiv} to replace each use of $\twistCTT\ccPath$ with $\ccPath$.
  Set $\twistCTT\ccEnglue(\sov{a},\sov{t}) \defeq \ccEnglue(\sov{a},\sov{t})$ and $\twistCTT\ccUnglue(\sov{g}) \defeq \ccUnglue(\sov{g})$.
\end{component}

The interpretation of suspensions is similar, but we use $\ccThicken$ and $\ccAnti$ more directly.

\begin{component}[$\twistCTT$, suspension]
  Define
  \[
    \begin{array}[t]{lcl}
      \twistCTT\ccSusp(\sov{A}) &\defeq& \ccSusp(\sov{A}) \\
      \twistCTT\ccNorth &\defeq& \ccNorth \\
      \twistCTT\ccSouth &\defeq& \ccSouth
    \end{array}
    \qquad
    \begin{array}[t]{lcl}
      \twistCTT\ccMerid(\sov{a}) &\defeq& \ccThicken(\ccMerid(\sov{a})) \\
      \twistCTT\ccSuspElim(\sov{C},\sov{n},\sov{s},\sov{m},\sov{t}) &\defeq& \ccSuspElim(\sov{C},\sov{n},\sov{s},\soabs{\sov{a}}\ccAnti(\sov{m}(\sov{a})),\sov{t})
    \end{array}
  \]
  For $\twistCTT\ccSuspElimMerid(\sov{C},\sov{n},\sov{s},\sov{m},\sov{a})$, we compose
  $\ccCong_\ccThicken(\ccSuspElimMerid(\sov{C},\sov{n},\sov{s},\soabs{\sov{a}}\ccThicken^{-1}(\sov{m}(\sov{a})),\sov{a}))$ with the path $\ccThicken(\ccThicken^{-1}(\sov{q})) \sim \sov{q}$, using that $\ccThicken$ is an equivalence, then thicken the composed path to get a $\twistCTT$-path.
\end{component}

This completes the definition of $\twistCTT$, as summarized in the following theorem.
We record that it preserves the constructs of $\rmcMLTT$ and the cofibration judgments for future use.

\begin{theorem}
  \label{twist-theorem}
  For every self-dual interval theory $(\Phi,\phi)$, there is a representable map functor $\twistCTT \colon \rmcCTT[\iota\thRev{\Phi}{\phi}] \to \rmcCTT[\iota\Phi]$ in the coslice $\Coslice{(\rmcMLTTU + \rmcCof)}{\RMC}$.
\end{theorem}

\section{Spans}
\label{sec:span}

Abstracting from the particular case of $\twistCTT$, we now develop tools---span RMCs and the \emph{span interpretation} between suitable RMC functors $F,G \colon \rmcCTT[\iota\Phi] \to \rmcCTT[\iota\Psi]$---that we use in \S\ref{sec:conservativity} to prove that certain morphisms of models induced by RMC functors are weak equivalences.
This construction at the level of RMCs is inspired by and resembles path object constructions at the level of models \cite[\S5]{kapulkin-lumsdaine:18}, as well as Tabareau, Tanter, and Sozeau's \emph{univalent parametricity} translation for the Calculus of Inductive Constructions \cite{tabareau-tanter-sozeau:21}.

\subsection{The representable map category of spans}

We write $\Span{\cat{C}}$ for the category of \emph{spans in a category $\cat{C}$}, \ie, the category of functors from the diagram category $\{ 0 \leftarrow \R \rightarrow 1 \}$ into $\cat{C}$.
Given $X \in \Span{\cat{C}}$, we write $\spanl \colon X_\R \to X_0$ and $\spanr \colon X_\R \to X_1$ for its two projections.

\begin{proposition}
  If $\rmc{R}$ is an RMC, then $\Span{\rmc{R}}$ is an RMC when equipped with the class of levelwise representable maps.
\end{proposition}
\begin{proof}
  As a functor category, $\rmc{R}$ has finite limits computed pointwise in $\rmc{R}$.
  In particular, the fact that representable maps are closed under pullback in $\rmc{R}$ implies the same of $\Span{\rmc{R}}$.

  It remains to show that the representable maps in $\Span{\rmc{R}}$ are exponentiable.
  Let $f \colon Z \to Y$ and $p \colon Y \to X$ be maps in $\Span{\rmc{R}}$ and suppose $p$ is representable.
  As $p_0$ and $p_1$ are exponentiable, we have dependent products $g_0 \defeq \ran{(p_0)}f_0 \colon \Pi_{p_0}Z_0 \to X_0$ and $g_1 \defeq \ran{(p_1)}f_1 \colon \Pi_{p_1}Z_1 \to X_1$.
  Write $k$ for the composite
  \[
    \begin{tikzcd}[column sep=8em]
      Y_\R \times_{X_0 \times X_1} \left(\Pi_{p_0}Z_0 \times \Pi_{p_1}Z_1\right) \ar{d}[sloped]{\cong} \\
      Y_\R \times_{Y_0 \times Y_1} \left((Y_0 \times_{X_0} \Pi_{p_0}Z_0) \times (Y_1 \times_{X_1} \Pi_{p_1}Z_1)\right) \ar{r}{Y_\R \times_{Y_0 \times Y_1} (\epsilon_{Z_0} \times \epsilon_{Z_1}))} &
      Y_\R \times_{Y_0 \times Y_1} Z_0 \times Z_1
    \end{tikzcd}
  \]
  induced by the counits of the (pullback, pushforward) adjunction, and write $q$ for the (representable) pullback
  \[
    \begin{tikzcd}
      Y_\R \times_{X_0 \times X_1} \left(\Pi_{p_0}Z_0 \times \Pi_{p_1}Z_1\right) \pullback \ar[fib,dashed]{d}[left]{q} \ar{r} & Y_\R \ar[fib]{d}{p_\R} \\
      X_\R \times_{X_0 \times X_1} \left(\Pi_{p_0}Z_0 \times \Pi_{p_1}Z_1\right) \ar{r} & X_\R \rlap{.}
    \end{tikzcd}
  \]
  Writing the components of
  $\ran{q}\subst{k}\pair{f_\R}{\pair{\spanl}{\spanr}} \colon \Pi_q \subst{k}Z_\R \to X_\R \times_{X_0 \times X_1} \left(\Pi_{p_0}Z_0 \times \Pi_{p_1}Z_1\right)$ as $\soabs{g_\R,\spanl,\spanr}$, the morphism of spans
  \[
    \begin{tikzcd}
      \Pi_{p_0} Z_0 \ar{d}[left]{g_0} & \ar{l}[above]{\spanl} \Pi_q \subst{k}Z_\R \ar{d}{g_\R} \ar{r}{\spanr} & \Pi_{p_1} Z_1 \ar{d}{g_1} \\
      X_0 & \ar{l}{\spanl} X_\R \ar{r}[below]{\spanr} & X_1
    \end{tikzcd}
  \]
  is a pushforward of $f$ along $p$.
\end{proof}

By definition, the projections $\pi_0,\pi_1 \colon \Span{\rmc{R}} \to \rmc{R}$ are RMC functors.
Restricting our attention now to $\rmcMLTT$, we define an RMC functor $\spanMLTT$ fitting in the diagram
\[
  \begin{tikzcd}
    & \rmcMLTT \ar[equals]{dl} \ar[dashed]{d}{\spanMLTT} \ar[equals]{dr} \\
    \rmcMLTT & \ar{l}[below]{\pi_0} \Span{\rmcMLTT} \ar{r}[below]{\pi_1} & \rmcMLTT
  \end{tikzcd}
\]
by giving an interpretation.
For $\Phi \in \rmcMLTT$, we write the span $\spanMLTT\Phi$ as
$\Phi \overset{\spanl_\Phi}{\leftarrow} \apexMLTT\Phi\overset{\spanr_\Phi}{\to} \Phi$, \ie, with $\apexMLTT \colon \rmcMLTT \to \rmcMLTT$ denoting the composite of $\spanMLTT$ with the apex projection.

To interpret the type and term judgments, we will use the environment $\ccTyEqv$ of 1-to-1 correspondences from \cref{correspondence}.

\begin{definition}
  Write $\ccTmEqv \defeq ((\sov{A},\sov{A}',\overline{\sov{A}},\cdots) : \ccTyEqv, \sov{a} : \sov{A}, \sov{a}' : \sov{A}', \overline{\sov{a}} : \overline{\sov{A}}(\sov{a},\sov{a}'))\in \rmcMLTT$
  and $\spanl,\spanr \colon \ccTmEqv \to \ccTm$ for the instantiations projecting $(\sov{A},\sov{a})$ and $(\sov{A}',\sov{a}')$ respectively.
\end{definition}

\begin{component}[$\spanMLTT$, sorts]
  For sorts, we define $\spanMLTT\ccTy \defeq \braces{\ccTy \overset{\spanl}\leftarrow \ccTyEqv \overset{\spanr}\to \ccTy}$ and $\spanMLTT\ccTm \defeq \braces{\ccTm \overset{\spanl}\leftarrow \ccTmEqv \overset{\spanr}\to \ccTm}$, with $\spanMLTT\clTm \colon \spanMLTT\ccTy \fib \spanMLTT\ccTm$ the evident projection.
\end{component}

Defining $\spanMLTT$ for a type former $T \colon \Phi \soto \ccTy$ now amounts to giving, over the environment $(\sov{p} : \apexMLTT\Phi)$, a 1-to-1 correspondence $\relMLTT(\sov{p},-,-)$ between $T(\spanl_\Phi(\sov{p}))$ and $T(\spanr_\Phi(\sov{p}))$.
Similarly, interpreting a term former $t \colon (\sov{x} : \Phi) \soto \ccTm(I(\sov{x}))$ amounts to giving over $(\sov{p} : \apexMLTT\Phi)$ an inhabitant of $\relMLTT(\sov{p},t(\spanl_\Phi(\sov{p})),t(\spanr_\Phi(\sov{p})))$.

\begin{component}[$\spanMLTT$, unit type]
  We interpret the unit type former by the 1-to-1 correspondence $\relMLTT\mathsf{Unit}(\_,\_,\_) = \mathsf{Unit}$ and its unique inhabitant by the unique witness.
\end{component}

\begin{component}[$\spanMLTT$, $\Sigma$ types]
  \label{span-sigma}
  In the environment consisting of $\sov{A} : \ccTy$, $\sov{A}' : \ccTy$, a 1-to-1 correspondence $\overline{\sov{A}} : (\sov{A}, \sov{A}') \to \ccTy$, families $\sov{B} : \sov{A} \to \ccTy$, $\sov{B'} : \sov{A}' \to \ccTy$, and a family of 1-to-1 correspondences
  $
    \overline{\sov{B}} : ([\sov{a} : \sov{A}, \sov{a}' : \sov{A}'], \overline{\sov{a}} : \overline{\sov{A}}(\sov{a},\sov{a}'), \sov{b} : \sov{B}(\sov{a}), \sov{b}' : \sov{B}'(\sov{a}')) \to \ccTy
  $,
  we define $\relMLTT\ccSigma$ at $\sov{s} : \ccSigma(A,B)$ and $\sov{s}' : \ccSigma(A',B')$ to be
  $\ccSigma\overline{\sov{a}}:\overline{\sov{A}}(\sov{s}.1,\sov{s}'.1).\ \overline{\sov{B}}(\overline{\sov{a}},\sov{s}.2,\sov{s}'.2)$.
  We interpret pairing and projection by pairing and projection in this $\Sigma$ type.
\end{component}

\begin{component}[$\spanMLTT$, identity types]
  In the environment consisting of $\sov{A} : \ccTy$, $\sov{A}' : \ccTy$, a 1-to-1 correspondence $\overline{\sov{A}} : (\sov{A}, \sov{A}') \to \ccTy$, terms $\sov{a}_0\ \sov{a}_1 : \sov{A}$ and $\sov{a}_0'\ \sov{a}_1' : \sov{A}'$, and $\overline{\sov{a}}_0 : \overline{\sov{A}}(\sov{a}_0,\sov{a}_0')$ and $\overline{\sov{a}}_1 : \overline{\sov{A}}(\sov{a}_1,\sov{a}_1')$,
  we define $\relMLTT\Id$ at $\sov{p} : \sov{a}_0 \asymp \sov{a}_1$ and $\sov{p}' : \sov{a}_0' \asymp \sov{a}_1'$ to be the type of identities between $\overline{\sov{a}}_0 : \overline{\sov{A}}(\sov{a}_0,\sov{a}_0')$ and $\overline{\sov{a}}_1 : \overline{\sov{A}}(\sov{a}_1,\sov{a}_1')$ over $\sov{p}$ and $\sov{p}'$, \ie, the type of identities between the transport of $\overline{\sov{a}}_0$ along these identities and $\overline{\sov{a}}_1$.
\end{component}

This completes the definition of $\spanMLTT \colon \rmcMLTT \to \Span{\rmcMLTT}$.
Given an RMC $i \colon \rmcMLTT \to \rmc{R}$ under $\rmcMLTT$, we can now regard $\Span{\rmc{R}}$ as an RMC under $\rmcMLTT$ by way of the composite $\Span{i} \circ \spanMLTT \colon \rmcMLTT \to \Span{\rmc{R}}$.
In particular, we have $\Synt{\Span{\rmc{R}}} \in \Mod{\rmcMLTT}$ as in \Cref{models-from-syntax}.

\begin{proposition}
  \label{span-legs-weak-equivalences}
  For any $i \colon \rmcMLTT \to \rmc{R}$, the morphisms $\Synt{\pi_0}, \Synt{\pi_1} \colon \Synt{\Span{\rmc{R}}} \to \Synt{\rmc{R}}$ in $\Mod{\rmcMLTT}$ are weak equivalences.
\end{proposition}
\begin{proof}
  We consider $\pi_0 \colon \Span{\rmc{R}} \to \rmc{R}$, the case of $\pi_1$ being symmetric.
  An object of $\Gamma \in \Ctx{\Synt{\Span{\rmc{R}}}}$ is a span $\braces{\Gamma_0 \overset{\spanl}\leftarrow \Gamma_\R \overset{\spanr}\to \Gamma_1}$ obtained by iterated context extension of the terminal span with respect to the representable map $\ccTmEqv \to \ccTyEqv$ as described in \Cref{democracy}.
  It follows that $\Gamma_0$, $\Gamma_1$, and $\Gamma_\R$ are contexts in $\rmc{R}$, \ie, environments of term hypotheses, and that we have instantiations $\unitl \colon \Gamma_0 \to \Gamma_\R$ and $\unitr \colon \Gamma_1 \to \Gamma_\R$ that are homotopy inverses of $\spanl \colon \Gamma_\R \to \Gamma_0$ and $\spanr \colon \Gamma_\R \to \Gamma_1$ at each entry.
  To show weak type lifting for $\pi_0 \colon \Span{\rmc{R}} \to \rmc{R}$, we are given $B \colon \Gamma_0 \to \ccTy$ in $\rmc{R}$ and must construct an $A \colon \Gamma \to \spanMLTT\ccTy$ in $\Span{\rmc{R}}$ for which $A_0$ is equivalent to $B$.
  We have
  \[
    \begin{tikzcd}
      \Gamma_0 \ar{d}[left]{B} & \ar{l}[above]{\spanl} \Gamma_\R \ar{r}{\spanr} & \Gamma_1 \ar{d}{B\spanl\unitr} \\
      \ccTy & \ccTyEqv \ar{l}{\spanl} \ar{r}[below]{\spanr} & \ccTy
    \end{tikzcd}
  \]
  and, since $\spanl\unitr\spanr$ is homotopic to $\spanl$ at each component, we can find a map $\Gamma_r \to \ccTyEqv$ that makes the diagram commute.
  Taking the result as our $A$, we have not only an equivalence from $A_0$ to $B$ but an equality.
  The construction of term lifting is analogous.
\end{proof}

\subsection{Relating interpretations using spans}

For this section, we fix two RMC functors $F,G \colon \rmcCTT[\iota\Phi] \to \rmcCTT[\iota\Psi]$ in the coslice under the combined sub-SOGAT $\rmcMLTTU + \rmcCof$, where $\rmcMLTTU$ is the extension of $\rmcMLTT$ by a universe $\ccU$ with $\ccEl$.
We construct a third RMC functor $\spanCTT{F}{G} \colon \rmcCTT[\iota\Phi] \to \Span{\rmcCTT[\iota\Psi]}$ in the coslice under $\rmcMLTT$ that fits in the diagram
\[
  \begin{tikzcd}
    & \rmcCTT[\iota\Phi] \ar{dl}[above left]{F} \ar[dashed]{d}{\spanCTT{F}{G}} \ar{dr}{G} \\
    \rmcCTT[\iota\Psi] & \ar{l}[below]{\pi_0} \Span{\rmcCTT[\iota\Psi]} \ar{r}[below]{\pi_1} & \rmcCTT[\iota\Psi] \rlap{.}
  \end{tikzcd}
\]
The effect of this functor, which we exploit in \S\ref{sec:conservativity}, is to show that $F$ and $G$ are approximately ``the same''.
Note that we need to know very little about $F$ and $G$ to obtain $\spanCTT{F}{G}$: this reflects that the constructs of $\rmcCTT[\iota\Phi]$ are all characterized up to equivalence by their universal properties, so an interpretation has little choice in where to send them.
It is key here that we are looking at \emph{second-order} models, \ie, RMC functors; we would not have the same result for morphisms of first-order models.

For $\Theta \in \rmcCTT$, we write the span $\spanCTT{F}{G}\Theta$ as
$F\Theta \overset{\spanl_\Theta}{\leftarrow} \apexCTT{F}{G}\Theta \overset{\spanr_\Theta}{\to} G\Theta$.
Because we intend $\spanCTT{F}{G}$ to be a morphism in the coslice under $\rmcMLTT$, the interpretations of the constructs of $\rmcMLTT$ are determined by the definition of $\spanMLTT$ from the previous section.

From this point until the summary statement \Cref{span-theorem}, we omit the annotations on $\spanCTT{F}{G}$ and $\apexCTT{F}{G}$ and simply write $\spanCTT*{F}{G}$ and $\apexCTT*{F}{G}$.

\begin{component}[$\spanCTT*{F}{G}$, sorts]
  \label{span-sorts}
  Set $\spanCTT*{F}{G}\ccTy \defeq \braces{\ccTy \overset{\spanl}\leftarrow \ccTyEqv \overset{\spanr}\to \ccTy}$ and $\spanCTT*{F}{G}\ccTm \defeq \braces{\ccTm \overset{\spanl}\leftarrow \ccTmEqv \overset{\spanr}\to \ccTm}$ as required by the definition of $\spanMLTT$.
  For the remaining sorts:
  \begin{enumerate}
  \item Set $\spanCTT*{F}{G}\ccII \defeq \braces{F\ccII \overset{\spanl}\leftarrow F\ccII \times G\ccII \overset{\spanr}\to G\ccII}$.
  \item Set $\apexCTT*{F}{G}\ccCof \defeq (\sov{P}\;\sov{P}'\;\overline{\sov{P}} : \ccCof, [\overline{\sov{P}} \to \sov{P}, \overline{\sov{P}} \to \sov{P}'])$ with $\spanl_\ccCof,\spanr_\ccCof$ projecting $\sov{P}$ and $\sov{P}'$ respectively.
  \item Set $\apexCTT*{F}{G}\ccTrue \defeq ((\sov{P},\sov{P}',\overline{\sov{P}}) : \apexCTT*{F}{G}\ccCof, \overline{\sov{P}})$ with  $\spanl_\ccTrue,\spanr_\ccTrue$ applying the implications $\overline{\sov{P}} \to \sov{P}$ and $\overline{\sov{P}} \to \sov{P}'$, and define $\apexCTT*{F}{G}\clTrue$ to be the evident projection.
  \end{enumerate}
\end{component}

\begin{component}[$\spanCTT*{F}{G}$, interval theory]
  By definition of $\spanCTT*{F}{G}\ccII$, the interpretation of the interval theory is forced by $F$ and $G$.
  Unfolding, the interpretation $\apexCTT*{F}{G}f$ of each operation $f \colon \ccII^n \to \II$ of the interval theory is $(F\ccII \times G\ccII)^n \cong F\ccII^n \times G\ccII^n \overset{Ff \times Gf}\longrightarrow F\II \times G\ccII$.
\end{component}

\begin{component}[$\spanCTT*{F}{G}$, cofibration theory]
  We interpret the cofibration operations as follows.
  \begin{align*}
    (\sov{i},{\sov{x}}) \ccCofEq_{\apexCTT*{F}{G}} (\sov{j},{\sov{y}}) &\defeq (\sov{i} \ccCofEq_F \sov{j}, {\sov{x}} \ccCofEq_G {\sov{y}}, (\sov{i} \ccCofEq_F \sov{j}) \ccCap ({\sov{x}} \ccCofEq_G {\sov{y}})) \\
    \apexCTT*{F}{G}\ccTop &\defeq (F\ccTop, G\ccTop, \ccTop) \\
    (\sov{P},\sov{P}',\overline{\sov{P}}) \ccCap_{\apexCTT*{F}{G}} (\sov{Q},\sov{Q}',\overline{\sov{Q}}) &\defeq (\sov{P} \ccCap_F \sov{Q}, \sov{P}' \ccCap_G \sov{Q}', \overline{\sov{P}} \ccCap \overline{\sov{Q}}) \\
    \apexCTT*{F}{G}\ccBot &\defeq (F\ccBot, G\ccBot, \ccBot) \\
    (\sov{P},\sov{P}',\overline{\sov{P}}) \ccCup_{\apexCTT*{F}{G}} (\sov{Q},\sov{Q}',\overline{\sov{Q}}) &\defeq (\sov{P} \ccCup_F \sov{Q}, \sov{P}' \ccCup_G \sov{Q}', \overline{\sov{P}} \ccCup \overline{\sov{Q}})
  \end{align*}
  The axioms for cofibrations ensure that these definitions preserve the implicit requirement that for $(\sov{P}, \sov{P}', \overline{\sov{P}}) : \apexCTT*{F}{G}\ccCof$ we have $\overline{\sov{P}} \to \sov{P}$ and $\overline{\sov{P}} \to \sov{P}'$.
  We use this to interpret the $\ccCaseTy$ and $\ccCaseTm$ eliminators.
  For $\ccCaseTy$, for example, we are given $(\sov{P},\sov{P}',\overline{\sov{P}}) : \apexCTT*{F}{G}\ccCof$, $(\sov{Q},\sov{Q}',\overline{\sov{Q}}) : \apexCTT*{F}{G}\ccCof$, compatible $\sov{A} : [\sov{P}] \to \ccTy$ and $\sov{B} : [\sov{Q}] \to \ccTy$, compatible $\sov{A}' : [\sov{P}'] \to \ccTy$ and $\sov{B}' : [\sov{Q}'] \to \ccTy$, and compatible 1-to-1 correspondences $\overline{\sov{A}} : ([\overline{\sov{P}}],\sov{A},\sov{A}') \to \ccTy$ and $\overline{\sov{B}} : ([\overline{\sov{Q}}],\sov{B},\sov{B}') \to \ccTy$, and we need to extend these to a 1-to-1 correspondence between $F\ccCaseTy(\sov{P},\sov{Q},\sov{A},\sov{B})$ and $G\ccCaseTy(\sov{P}',\sov{Q}',\sov{A}',\sov{B}')$ assuming $\overline{\sov{P}} \ccCup \overline{\sov{Q}}$.
  To do so we case on $\overline{\sov{P}} \ccCup \overline{\sov{Q}}$ and use that we either have both $\sov{P}$ and $\sov{P}'$ or both $\sov{Q}$ and $\sov{Q}'$ as a consequence.
\end{component}

\begin{component}[$\spanCTT*{F}{G}$, filling]
  To define $\spanCTT*{F}{G}\ccFill$, we are given inputs
  \[
    \begin{array}[t]{lcl}
      \sov{A} &:& F\ccII \to \ccTy \\
      \sov{A}' &:& G\ccII \to \ccTy \\
      \overline{\sov{A}} &:& (\sov{i} : F\ccII, {\sov{x}} : G\ccII, \sov{a} : \sov{A}(\sov{i}), \sov{a}' : \sov{A}'({\sov{x}})) \to \ccTy \\
      (\sov{P},\sov{P}',\overline{\sov{P}}) &:& \apexCTT*{F}{G}\ccCof \\
      \sov{a} &:& (\sov{i} : F\ccII, \sov{P}) \to \sov{A}(\sov{i}) \\
      \sov{a}' &:& ({\sov{x}} : G\ccII, \sov{P}') \to \sov{A}'({\sov{x}}) \\
      \overline{\sov{a}} &:& (\sov{i} : F\ccII, {\sov{x}} : G\ccII, \overline{\sov{P}}) \to \overline{\sov{A}}(\sov{i},{\sov{x}},\sov{a}(\sov{i}),\sov{a}'({\sov{x}}))
    \end{array}
    \qquad
    \begin{array}[t]{lcl}
      (\sov{j},{\sov{y}}) &:& \apexCTT*{F}{G}\ccII \\
      \sov{a}_0 &:& \sov{A}(\sov{j}) \\
      \sov{a}'_0 &:& \sov{A}'({\sov{y}}) \\
      \overline{\sov{a}}_0 &:& \overline{\sov{A}}(\sov{j},{\sov{y}},\sov{a}_0,\sov{a}'_0) \\
      (\sov{k},{\sov{z}}) &:& \apexCTT*{F}{G}\ccII
    \end{array}
  \]
  satisfying $\sov{P} \to \sov{a}(\sov{j}) \equiv \sov{a}_0 : \sov{A}(\sov{j})$, $\sov{P}' \to \sov{a}'({\sov{y}}) \equiv \sov{a}'_0 : \sov{A}'({\sov{y}})$, and $\overline{\sov{P}} \to \overline{\sov{a}}(\sov{j},{\sov{y}}) \equiv \overline{\sov{a}}_0 : \overline{\sov{A}}(\sov{a}_0,\sov{a}'_0)$.
  Abbreviating $\sov{a}_+(\sov{k}) \defeq F\ccFill^{\sov{j}\to\sov{k}}(\sov{A},[\sov{P} \mapsto \sov{a}],\sov{a}_0)$ and $\sov{a}'_+(\sov{z}) \defeq G\ccFill^{\sov{y}\to\sov{z}}(\sov{A}',[\sov{P}' \mapsto \sov{a}'],\sov{a}'_0)$, we must exhibit a term of type $\overline{\sov{A}}(\sov{k},{\sov{z}},\sov{a}_+(\sov{k}),\sov{a}'_+({\sov{z}}))$.
  We take the iterated filling expression
  \[
    \def\arraystretch{1.2}
    \def\arraycolsep{0pt}
    G\ccFill^{{\sov{y}}\to{\sov{z}}}
    \begin{array}[t]{l}
      (\soabs{{\sov{x}}}\overline{\sov{A}}(\sov{k},{\sov{x}},\sov{a}_+(\sov{k}),\sov{a}'_+({\sov{x}})), [\overline{\sov{P}} \mapsto \soabs{{\sov{x}}}\overline{\sov{a}}(\sov{k},{\sov{x}})], \\
      \phantom{(}F\ccFill^{\sov{j}\to\sov{k}}(\soabs{\sov{i}}\overline{\sov{A}}(\sov{i},{\sov{y}},\sov{a}_+(\sov{i}),\sov{a}'_0), [\overline{\sov{P}} \mapsto \soabs{\sov{i}}\overline{\sov{a}}(\sov{i},{\sov{y}})], \overline{\sov{a}}_0)) \rlap{.}
    \end{array}
  \]
\end{component}

\begin{component}[$\spanCTT*{F}{G}$, $\Pi$ types]
  \label{span-pi}
  In the environment consisting of $\sov{A} : \ccTy$, $\sov{A}' : \ccTy$, a 1-to-1 correspondence $\overline{\sov{A}} : (\sov{A}, \sov{A}') \to \ccTy$, families $\sov{B} : \sov{A} \to \ccTy$, $\sov{B'} : \sov{A}' \to \ccTy$, and 1-to-1 correspondences
  $
    \overline{\sov{B}} : ([\sov{a} : \sov{A}, \sov{a}' : \sov{A}'], \overline{\sov{a}} : \overline{\sov{A}}(\sov{a},\sov{a}'), \sov{b} : \sov{B}(\sov{a}), \sov{b}' : \sov{B}'(\sov{a}')) \to \ccTy
  $,
  we take the relation sending $\sov{f} : F\ccPi(\sov{A},\sov{B})$ and $\sov{f}' : G\ccPi(\sov{A}',\sov{B}')$ to
  $\ccPi\sov{a}:\sov{A}.\ \ccPi\sov{a}':\sov{A}'.\ \ccPi\overline{\sov{a}}:\overline{\sov{A}}(\sov{a},\sov{a}').\ \overline{\sov{B}}(\overline{\sov{a}},\sov{f}(\sov{a}),\sov{f}'(\sov{a}'))$.
\end{component}

For $\ccPath$ types, we exploit the fact that we can convert between non-dependent $\ccPath$, $F\ccPath$, and $G\ccPath$ types, which follows from the fact that both types support coercion.

\begin{lemma}
  \label{decode-path}
  Over $([\sov{C} : \ccTy, \sov{c}_0 : \sov{C}$, $\sov{c}_1 : \sov{C}], \sov{p} : F\ccPath(\soabs{\_}\sov{C}, \sov{c}_0, \sov{c}_1))$, we have a term $\ccDecode{F}(\sov{p}) : \sov{c}_0 \sim^{\sov{C}} \sov{c_1}$.
\end{lemma}
\begin{proof}
  We have $F\ccCoe : (\sov{A} : ({\sov{x}} : F\ccII) \to \ccTy, {\sov{y}} : F\ccII, \sov{a}_0 : \sov{A}({\sov{y}}), {\sov{z}} : F\ccII) \soto \sov{A}({\sov{z}})$, and instantiating with the arguments $(\soabs{{\sov{x}}}(\sov{c}_0 \sim^{\sov{C}} \sov{p} \ccPathApp_F {\sov{x}}), F\cce0, (\ccPathLam*\_. \sov{c}_0), F\cce1)$ yields the desired expression.
\end{proof}

\begin{corollary}
  \label{F-singleton-contractibility}
  Over $(\sov{C} : F\ccII \to \ccTy, \sov{c}_0 : \sov{C}(F\cce0))$, $\ccSigma \sov{c}_1{:}\sov{C}(F\cce1). F\ccPath(\sov{C},\sov{c}_0,\sov{c}_1)$ is contractible.
\end{corollary}
\begin{proof}
  Applying $F$ to singleton contractibility (\Cref{singleton-contractibility}), we get over the environment $(\sov{C} : F\ccII \to \ccTy, \sov{c}_0 : \sov{C}(F\cce0))$ a term of type
  $
    F\ccIsContr(F\ccSigma \sov{c}_1{:}\sov{C}(F\cce1). F\ccPath(\sov{C},\sov{c}_0,\sov{c}_1))\rlap{.}
  $
  The type formers $F\Sigma$ and $F\Pi$ satisfy the rules for $\Sigma$- and $\Pi$-types, so we can define equivalences $F\ccSigma(A,B) \simeq \ccSigma(A,B)$ and $F\ccPi(A,B) \simeq \ccPi(A,B)$. Combined with \Cref{decode-path}, we can therefore derive $\ccIsContr(\ccSigma \sov{c}_1{:}\sov{C}(F\cce1). F\ccPath(\sov{C},\sov{c}_0,\sov{c}_1))$.
\end{proof}

Of course, \Cref{F-singleton-contractibility} also holds when we replace $F$ with $G$.

\begin{component}[$\spanCTT*{F}{G}$, path types]
  \label{span-paths}
  To define $\spanCTT*{F}{G}\ccPath$, we are given $\sov{A} : F\II \to \ccTy$ and $\sov{A}' : G\II \to \ccTy$ with a 1-to-1 correspondence $\overline{\sov{A}} : (\sov{i} : F\ccII, {\sov{x}} : G\ccII, \sov{a} : \sov{A}(\sov{i}), \sov{a}' : \sov{A}'({\sov{x}})) \to \ccTy$, terms $\sov{a}_0 : \sov{A}(F\cce0)$ and $\sov{a}'_0 : \sov{A}'(G\cce0)$ with $\overline{\sov{a}}_{00} : \overline{\sov{A}}(F\cce0,G\cce0,\sov{a}_0,\sov{a}'_0)$, and terms $\sov{a}_1 : \sov{A}(F\cce1)$ and $\sov{a}'_1 : \sov{A}'(G\cce1)$ with $\overline{\sov{a}}_{11} : \overline{\sov{A}}(F\cce1,G\cce1,\sov{a}_1,\sov{a}'_1)$.

  We need to define a 1-to-1 correspondence between $F\ccPath(\sov{A},\sov{a}_0,\sov{a}_1)$ and $G\ccPath(\sov{A}',\sov{a}'_0,\sov{a}'_1)$.
  We take the relation sending $\sov{p}$ and $\sov{p}'$ to the iterated $\Sigma$-type with components
  \begin{equation}
    \label{span-path-sigma}
    \def\arraystretch{1.2}
    \begin{array}{lcl}
      \overline{\sov{a}}_{10} : \overline{\sov{A}}(F\cce1,G\cce0,\sov{a}_1,\sov{a}'_0). \\
      \overline{\sov{a}}_{01} : \overline{\sov{A}}(F\cce0,G\cce1,\sov{a}_0,\sov{a}'_1). \\
      \overline{\sov{a}}_{\bullet0} : F\ccPath(\soabs{\sov{i}}\overline{\sov{A}}(\sov{i},G\cce0,\sov{p} \ccPathApp_F\sov{i}, \sov{a}'_0), \overline{\sov{a}}_{00}, \overline{\sov{a}}_{10}). \\
      \overline{\sov{a}}_{\bullet1} : F\ccPath(\soabs{\sov{i}}\overline{\sov{A}}(\sov{i},G\cce1,\sov{p} \ccPathApp_F\sov{i}, \sov{a}'_1), \overline{\sov{a}}_{01}, \overline{\sov{a}}_{11}). \\
      \overline{\sov{a}}_{0\bullet} : G\ccPath(\soabs{{\sov{x}}}\overline{\sov{A}}(F\cce0,{\sov{x}},\sov{a}_0, \sov{p}' \ccPathApp_G {\sov{x}}), \overline{\sov{a}}_{00}, \overline{\sov{a}}_{01}). \\
      \overline{\sov{a}}_{1\bullet} : G\ccPath(\soabs{{\sov{x}}}\overline{\sov{A}}(F\cce1,{\sov{x}},\sov{a}_1, \sov{p}' \ccPathApp_G {\sov{x}}), \overline{\sov{a}}_{10}, \overline{\sov{a}}_{11}). \\
      \def\arraycolsep{0pt}
      \overline{\sov{a}}_{\bullet\bullet} : F\ccPath(\soabs{\sov{i}}G\ccPath(\soabs{{\sov{x}}}\overline{\sov{A}}(\sov{i},{\sov{x}},\sov{p} \ccPathApp_F\sov{i}, \sov{p}' \ccPathApp_G {\sov{x}}), \overline{\sov{a}}_{\bullet0} \ccPathApp_F\sov{i}, \overline{\sov{a}}_{\bullet1} \ccPathApp_F\sov{i}), \overline{\sov{a}}_{0\bullet}, \overline{\sov{a}}_{1\bullet})\rlap{.}
    \end{array}
  \end{equation}
  An element consists effectively of a family of witnesses $\overline{\sov{a}}_{\bullet\bullet} \ccPathApp_F \sov{i} \ccPathApp_G {\sov{x}} : \overline{\sov{A}}(\sov{i},{\sov{x}},\sov{p} \ccPathApp_F \sov{i}, \sov{p}' \ccPathApp_G {\sov{x}})$
  satisfying $\overline{\sov{a}}_{\bullet\bullet} \ccPathApp_F F\cce0 \ccPathApp_G G\cce0 \equiv \overline{\sov{a}}_{00}$ and $\overline{\sov{a}}_{\bullet\bullet} \ccPathApp_F F\cce1 \ccPathApp_G G\cce1 \equiv \overline{\sov{a}}_{11}$.
  We define $\apexCTT*{F}{G}\ccPathLam$ and $\ccPathApp_{\apexCTT*{F}{G}}$ to be abstraction and application of such families.

  It remains to check that this relation is a 1-to-1 correspondence.
  Fix $\sov{p} : F\ccPath(\sov{A},\sov{a}_0,\sov{a}_1)$ and consider the type of pairs of $\sov{p}' : G\ccPath(\sov{A}',\sov{a}'_0,\sov{a}'_1)$ with the data \eqref{span-path-sigma}.
  Given the preceding data, the types of pairs $(\overline{\sov{a}}_{10},\overline{\sov{a}}_{\bullet 0})$ and $(\overline{\sov{a}}_{1\bullet}, \overline{\sov{a}}_{\bullet\bullet})$ as in \eqref{span-path-sigma} are dependent $F$-singletons, so contractible by \Cref{F-singleton-contractibility}.
  The type of pairs $(\overline{\sov{a}}_{01},\overline{\sov{a}}_{0\bullet})$ is likewise a dependent $G$-singleton and thus contractible.
  After contracting all of these, we are left with $\sov{p}' : G\ccPath(\sov{A}',\sov{a}'_0,\sov{a}'_1)$ and $\overline{\sov{a}}_{0\bullet} : G\ccPath(\soabs{{\sov{x}}}\overline{\sov{A}}(\cce0,{\sov{x}},\sov{a}_0, \sov{p}' \ccPathApp_G {\sov{x}}), \overline{\sov{a}}_{00}, \widetilde{a}_{01})$ where $\widetilde{a}_{01}$ is some expression.
  The type of such pairs is equivalent to $G\ccPath(\soabs{{\sov{x}}}\ccSigma\sov{a}':\sov{A}'({\sov{x}}).\overline{\sov{A}}(\cce0,{\sov{x}},\sov{a}_0,\sov{a}'), \pair{\sov{a}'_0}{\overline{\sov{a}}_{00}}, \pair{\sov{a}'_1}{\widetilde{a}_{01}})$, which is a $G\ccPath$-type over a contractible type and thus contractible.
  A symmetric argument deals with the case where we fix $\sov{p}'$ and allow $\sov{p}$ to vary freely.
\end{component}

\begin{component}[$\spanCTT*{F}{G}$, universes]
  Using the assumption that $F\ccU = G\ccU = \ccU$, we interpret the universe $\ccU$ by the relation sending $\sov{A} : \ccU$ and $\sov{A}' : \ccU$ to the type of $\ccU$-valued 1-to-1 correspondences between $\sov{A}$ and $\sov{A}'$.
  That this relation is itself a 1-to-1 correspondence is a consequence of univalence of $\ccU$ \cite[Theorem 5.8.4(iv)$\Rightarrow$(v)]{hott-book}.
  We interpret $\ccEl$, again using the assumption $F\ccEl(\sov{A}) = G\ccEl(\sov{A}) = \ccEl(\sov{A})$, as extracting the 1-to-1 correspondence.
\end{component}

\begin{component}[$\spanCTT*{F}{G}$, glue]
  \label{span-glue}
  To define $\spanCTT*{F}{G}\ccGlue$, we are given inputs
  \[
    \begin{array}[t]{lcl}
      (\sov{A},\sov{A}',\overline{\sov{A}}) &:& \apexCTT*{F}{G}\ccTy \\
      (\sov{P},\sov{P}',\overline{\sov{P}}) &:& \apexCTT*{F}{G}\ccCof \\
      \sov{T} &:& [\sov{P}] \to \ccTy \\
      \sov{T}' &:& [\sov{P}'] \to \ccTy \\
      \overline{\sov{T}} &:& ([\overline{\sov{P}}], \sov{t} : \sov{T}, \sov{t}' : \sov{T}') \to \ccTy \\
    \end{array}
    \qquad
    \begin{array}[t]{lcl}
      \sov{e} &:& [\sov{P}] \to \sov{T} \simeq_F \sov{A} \\
      \sov{e}' &:& [\sov{P}'] \to \sov{T}' \simeq_G \sov{A}' \\
      \overline{\sov{e}} &:& [\overline{\sov{P}}] \to \mathsf{R}_\simeq(\overline{\sov{A}},\overline{\sov{T}},\sov{e},\sov{e'}) \\
    \end{array}
  \]
  where $\overline{\sov{A}}$ and $\overline{\sov{T}}$ are 1-to-1 correspondences and $\mathsf{R}_{\simeq}(\overline{\sov{A}},\overline{\sov{T}},-,-)$ is the 1-to-1 correspondence between $\sov{T} \simeq_F \sov{A}$ and $\sov{T}' \simeq_G \sov{A}'$ given by the span interpretation of $(- \simeq -)$ at $\overline{\sov{T}}$ and $\overline{\sov{A}}$.
  We need to define a 1-to-1 correspondence between $F\ccGlue(\sov{A},\sov{P},\sov{T},\sov{e})$ and $G\ccGlue(\sov{A}',\sov{P}',\sov{T}',\sov{e}')$.
  We take the relation sending $\sov{g} : F\ccGlue(\sov{A},\sov{P},\sov{T},\sov{e})$ and $\sov{g'} : G\ccGlue(\sov{A}',\sov{P}',\sov{T}',\sov{e}')$ to
  \[
    \ccGlue(\overline{\sov{A}}(F\ccUnglue(\sov{g}),G\ccUnglue(\sov{g}')), \overline{\sov{P}}, \overline{\sov{T}}(\sov{g},\sov{g}'), \widehat{\mathsf{e}})
  \]
  where it remains to define $\widehat{\mathsf{e}} : \overline{\sov{T}}(\sov{g},\sov{g}') \simeq \overline{\sov{A}}(F\ccUnglue(\sov{g}),G\ccUnglue(\sov{g}'))$ under $\overline{\sov{P}}$.

  By the reduction equations for $F\ccUnglue(\sov{g})$ and $G\ccUnglue(\sov{g}')$ under $\sov{P}$ and $\sov{P}'$, the type for $\widehat{\mathsf{e}}$ simplifies to $\overline{\sov{T}}(\sov{g},\sov{g}') \simeq \overline{\sov{A}}(\sov{e}.1(\sov{g}),\sov{e}'.1(\sov{g}'))$.
  Per the interpretations of $\Sigma$ and $\Pi$ (\cref{span-sigma,span-pi}), $\overline{\sov{e}}$ contains a map $(\sov{t} : \sov{T}, \sov{t}' : \sov{T}') \to \overline{\sov{T}}(\sov{t},\sov{t}') \to \overline{\sov{A}}(\sov{e}.1(\sov{t}),\sov{e}'.1(\sov{t}'))$ as its first component.
  We take this map, instantiated at $\sov{g}$ and $\sov{g}'$, as the forward function of $\widehat{\mathsf{e}}$.
  To see that it is an equivalence, it suffices \cite[Theorem 11.1.6]{rijke:25} to check that the induced map on total spaces $(\Sigma\sov{t}:\sov{T}.\overline{\sov{T}}(\sov{t},\sov{g}')) \to (\Sigma\sov{a}:\sov{A}.\overline{\sov{A}}(\sov{a},\sov{e}'.1(\sov{g}')))$ is an equivalence, as the base map $\sov{e}.1 \co \sov{T} \to \sov{A}$ is an $F$-equivalence and thus an equivalence.
  This is the case because $\overline{\sov{A}}$ and $\overline{\sov{T}}$ are 1-to-1 correspondences and thus both sides are contractible.

  With this interpretation of $\ccGlue$, we can give the interpretations of $\ccEnglue$ and $\ccUnglue$ as $\ccEnglue$ and $\ccUnglue$.
\end{component}

To interpret suspension, we make essential use of identity types.

\begin{definition}
  Over the environment $([\sov{A} : \ccTy, \sov{A}' : \ccTy], \sov{f} : \sov{A} \to \sov{A}')$, define the type-valued relation
  $\ccGraph(\sov{f}) \defeq \soabs{\sov{a},\sov{a}'}(\sov{f}(\sov{a}) \asymp \sov{a}') : (\sov{a} : \sov{A}, \sov{a}' : \sov{A}') \to \ccTy$.
\end{definition}

For a map $\sov{f}$ that is an equivalence, $\ccGraph(\sov{f})$ is a 1-to-1 correspondence.
Conversely, a 1-to-1 correspondence $\overline{\sov{A}}$ between $\sov{A}$ and $\sov{A}'$ contains a map $\ccFwd_{\overline{\sov{A}}} : \sov{A} \to \sov{A}'$ that is an equivalence.

Over the environment $([\sov{A} : \ccTy, \sov{A}' : \ccTy], \sov{f} : \sov{A} \to \sov{A}')$, define $\ccSuspMap(\sov{f}) : F\ccSusp(\sov{A}) \to G\ccSusp(\sov{A}')$ by $F\ccSusp$-elimination so that $\ccSuspMap(\sov{f})(F\ccNorth) = G\ccNorth$, $\ccSuspMap(\sov{f})(F\ccSouth) = G\ccSouth$, and $\ccCong_{\ccSuspMap(\sov{f})}(F\ccMerid(\sov{a})) \sim G\ccMerid(\sov{a}')$.
If $f : A \to A'$ is an equivalence, then $\ccSuspMap(f)$ is an equivalence, by the elimination principles for $F\ccSusp(A)$ and $G\ccSusp(A')$.

\begin{component}[$\spanCTT*{F}{G}$, suspension]
  Over an environment with $\sov{A}\;\sov{A}' : \ccTy$ and a 1-to-1 correspondence $\overline{\sov{A}} \colon (\sov{A},\sov{A}') \to \ccTy$, we must construct a 1-to-1 correspondence between $F\ccSusp(\sov{A})$ and $G\ccSusp(\sov{A}')$; we take $\ccGraph(\ccSuspMap(\ccFwd_{\overline{\sov{A}}}))$.
  We interpret $\ccNorth$ and $\ccSouth$ by the reflexive identities $\ccSuspMap(\ccFwd_{\overline{\sov{A}}})(\ccNorth) \asymp \ccNorth$ and $\ccSuspMap(\ccFwd_{\overline{\sov{A}}})(\ccSouth) \asymp \ccSouth$.

  To interpret $\ccMerid$ applied to $\sov{a} : \sov{A}$, $\sov{a}' : \sov{A}'$, and $\overline{\sov{a}} : \overline{\sov{A}}(\sov{a},\sov{a}')$,  we first convert the path $\ccCong_{\ccSuspMap(\ccFwd_{\overline{\sov{A}}})}(F\ccMerid(\sov{a})) \sim G\ccMerid(\ccFwd_{\overline{\sov{A}}}(\sov{a}))$ to an identity, then rewrite along the identity $\ccFwd_{\overline{\sov{A}}}(\sov{a}) \asymp \sov{a}'$ obtained from $\overline{\sov{a}}$ to get an identity $\ccCong_{\ccSuspMap(\ccFwd_{\overline{\sov{A}}})}(F\ccMerid(\sov{a})) \asymp G\ccMerid(\sov{a}')$.
  Using $\apexCTT*{F}{G}\ccPathLam$, we convert this to an $\apexCTT*{F}{G}$-path in the necessary identity type.

  Now we interpret the eliminator.
  Over the environment
  \[
    (\sov{A} : \ccTy, \sov{C} : \ccSusp(\sov{A}) \to \ccTy, \sov{n} : \sov{C}(\ccNorth), \sov{s} : \sov{C}(\ccSouth), \sov{m} : (\sov{a} : \sov{A}) \to \ccPath(\soabs{\sov{i}}\sov{C}(\ccMerid(\sov{a}) \ccPathApp \sov{i}), \sov{n}, \sov{s}))
  \]
  we have a type
  \[
    \begin{array}{lcl}
      \mathsf{D} &\defeq& \ccSigma \sov{f}:(\Pi\sov{t}{:}\ccSusp(\sov{A}). \sov{C}(\sov{t})).\ \ccSigma \sov{p}_{\sov{n}} : \sov{f}(\ccNorth) \sim \sov{n}.\ \ccSigma \sov{p}_{\sov{s}} : \sov{f}(\ccSouth) \sim \sov{s}. \\
                 && \ccPath(\def\arraycolsep{0pt}\begin{array}[t]{l}\soabs{\sov{j}}\ccPath(\soabs{\sov{i}}\sov{C}(\ccMerid(\sov{a}) \ccPathApp \sov{i}),\sov{p}_{\sov{n}} \ccPathApp \sov{j},\sov{p}_{\sov{s}} \ccPathApp \sov{j}), (\ccPathLam*\sov{i}.\sov{f}(\ccMerid(\sov{a}) \ccPathApp \sov{i})),\sov{m})
                   \end{array}
    \end{array}
  \]
  of dependent functions into $\sov{C}$ defined on the constructors $\ccNorth$, $\ccSouth$, and $\ccMerid$ by $\sov{n}$, $\sov{s}$, and $\sov{m}$ respectively, up to homotopy.
  By virtue of the eliminator, $\mathsf{D}$ is contractible, as are $F\mathsf{D}$ and $G\mathsf{D}$.
  Thus every pair of elements from $F\mathsf{D}$ and $G\mathsf{D}$ is related in $\spanCTT*{F}{G}\mathsf{D}$, in particular the pair obtained from $F\ccSuspElim$ and $G\ccSuspElim$.
  This gives us an almost-interpretation of $\ccSuspElim$: we have an eliminator that may satisfy the point constructor computation rules only up to paths.

  We then correct our almost-interpretation on the point constructors.
  To interpret the eliminator, we are given type families $\sov{C}$ and $\sov{C}'$ and, over
  $(\sov{t} : F\ccSusp(\sov{A}), \sov{t}' : G\ccSusp(\sov{A}'), \overline{\sov{t}} : \ccSuspMap(\ccFwd_{\overline{\sov{A}}})(\sov{t}) \asymp \sov{t}')$, 1-to-1 correspondences
  $\overline{\sov{C}}(\sov{t}, \sov{t}', \overline{\sov{t}}) : (\sov{C}(\sov{t}), \sov{C}'(\sov{t}')) \to \ccTy$.
  We want to relate $F\ccSuspElim(\sov{C},\sov{n},\sov{s},\sov{m},\sov{t})$ and $G\ccSuspElim(\sov{C}',\sov{n}',\sov{s}',\sov{m}',\sov{t}'))$ in $\overline{\sov{C}}(\sov{t},\sov{t}',\overline{\sov{t}})$ for all related inputs.
  We go by $F\ccSusp$-elimination from $\sov{t}$ and identity elimination from $\overline{\sov{t}}$.
  For the point cases $\sov{t} = F\ccNorth$ and $\sov{t} = F\ccSouth$, we choose the values that the point computation rules require.
  For the $F\ccMerid$ case, we apply the almost-eliminator to the input data to get a section of $\overline{\sov{C}}$, evaluate it at the corresponding meridian, then coerce the result along the almost-eliminator's point computation paths to get a path of the correct type.
\end{component}

This completes the definition of $\spanCTT{F}{G} \colon \rmcCTT[\iota\Phi] \to \Span{\rmcCTT[\iota\Psi]}$.
In summary:

\begin{theorem}
  \label{span-theorem}
  Let $F,G \colon \rmcCTT[\iota\Phi] \to \rmcCTT[\iota\Psi]$ in the coslice under $\rmcMLTTU + \rmcCof$.
  There is a $\spanCTT{F}{G} \colon \rmcCTT[\iota\Phi] \to \Span{\rmcCTT[\iota\Psi]}$ in $\Coslice{\rmcMLTT}{\RMC}$ such that $\pi_0\spanCTT{F}{G} \cong F$ and $\pi_1\spanCTT{F}{G} \cong G$.
\end{theorem}

\section{Conservativity}
\label{sec:conservativity}

\begin{proposition}[2-out-of-6]
  \label{2-out-of-6}
  Weak equivalences of democratic models of $\rmcMLTT$ are closed under 2-out-of-6.
  That is, given morphisms of democratic models of $\rmcMLTT$
  \begin{tikzcd}[cramped]
    \mod{M}
    \ar[r, "\mod{F}"]
  &
    \mod{N}
    \ar[r, "\mod{G}"]
  &
    \mod{O}
    \ar[r, "\mod{H}"]
  &
    \mod{P}
  \end{tikzcd}
  where $\mod{G}\mod{F}$ and $\mod{H}\mod{G}$ are weak equivalences, the maps $\mod{F}$, $\mod{G}$, $\mod{H}$, and the composite $\mod{H}\mod{G}\mod{F}$ are weak equivalences.
\end{proposition}
\begin{proof}
  See Kapulkin and Lumsdaine \cite[Corollary 3.4]{kapulkin-lumsdaine:18}.
\end{proof}

A corollary of 2-out-of-6 is 2-out-of-3: given composable morphisms $\mod{G}$ and $\mod{F}$ between democratic models of $\rmcMLTT$, if two of the three morphisms $\mod{F}$, $\mod{G}$, and $\mod{G}\mod{F}$ are weak equivalences, then so is the third.

\begin{theorem}
  \label{rmc-biimplication-to-weak-equivalance}
  For $F \colon \rmcCTT[\iota\Phi] \to \rmcCTT[\iota\Psi]$ and $G \colon \rmcCTT[\iota\Psi] \to \rmcCTT[\iota\Phi]$ in the coslice under $\rmcMLTT + \rmcCof$, the induced morphisms
  $\Synt{F} \colon \Synt{\rmcCTT[\iota\Phi]} \to \Synt{\rmcCTT[\iota\Psi]}$ and $\Synt{G} \colon \Synt{\rmcCTT[\iota\Psi]} \to \Synt{\rmcCTT[\iota\Phi]}$ are weak equivalences.
\end{theorem}
\begin{proof}
  By \Cref{span-theorem}, we have an RMC functor fitting in the diagram in $\Coslice{\rmcMLTT}{\RMC}$ to the left below.
  \[
    \begin{tikzcd}[column sep=small]
      & \rmcCTT[\iota\Phi] \ar{dl}[above left]{GF} \ar[dashed]{d}{\spanCTT{GF}{\Id}} \ar[equal]{dr} \\
      \rmcCTT[\iota\Phi] & \ar{l}[below]{\pi_0} \Span{\rmcCTT[\iota\Phi]} \ar{r}[below]{\pi_1} & \rmcCTT[\iota\Phi]
    \end{tikzcd}
    \qquad
    \begin{tikzcd}[column sep=small]
      & \Synt{\rmcCTT[\iota\Phi]} \ar{dl}[above left]{\Synt{GF}} \ar{d}{} \ar[equal]{dr} \\
      \Synt{\rmcCTT[\iota\Phi]} & \ar{l}[below]{\sim} \Synt{\Span{\rmcCTT[\iota\Phi]}} \ar{r}[below]{\sim} & \Synt{\rmcCTT[\iota\Phi]}
    \end{tikzcd}
  \]
  This induces a diagram in $\Mod{\rmcMLTT}$ as shown to the right, where the morphisms marked $\sim$ are weak equivalences by \Cref{span-legs-weak-equivalences}.
  By two applications of 2-out-of-3, first in the right triangle and then in the left, it follows that $\Synt{GF}$ is a weak equivalence.

  By the same argument, $\Synt{FG} \colon \Synt{\rmcCTT[\iota\Psi]} \to \Synt{\rmcCTT[\iota\Psi]}$ is a weak equivalence.
  The claim now follows by 2-out-of-6 applied to the string of morphisms $\Synt{F} \circ \Synt{G} \circ \Synt{F}$.
\end{proof}

\begin{theorem}[Conservativity of reversals]
  \label{reversal-conservativity}
  For every self-dual interval theory $(\Phi,\phi)$, the inclusion $\rmcCTT[\iota\Phi] \to \rmcCTT[\iota\thRev{\Phi}{\phi}]$ induces a weak equivalence $\Synt{\rmcCTT[\iota\Phi]} \to \Synt{\rmcCTT[\iota\thRev{\Phi}{\phi}]}$ in $\Mod{\rmcMLTT}$.
\end{theorem}
\begin{proof}
  By \Cref{rmc-biimplication-to-weak-equivalance} with \Cref{twist-theorem}.
\end{proof}

\section{Interpreting strict cubical type theory with reversals in spaces}
\label{sec:spaces}

Kapulkin and Lumsdaine \cite{kapulkin-lumsdaine:18} show that every democratic model $\mod{M} \in \Mod{\rmcMLTT}$ induces a \emph{fibration category} structure on its category of contexts $\Ctx{\mod{M}}$.
Such a structure, which is specified by two classes of morphisms in $\Ctx{\mod{M}}$ called \emph{fibrations} and \emph{weak equivalences}, induces in turn an \emph{$(\infty,1)$-category} \cite{szumilo:16} or ``homotopy theory''.
It is in this way that we judge the kind of higher structure described by a model of $\rmcMLTT$.
The homotopy theory of topological spaces corresponds to one such $(\infty,1)$-category, that of \emph{$\infty$-groupoids}.

Awodey et al.~\cite{accrs:24} and Cavallo and Sattler \cite{cavallo-sattler:25} exhibit constructive models $\mod{M}$ of strict cubical type theories without reversals whose induced $(\infty,1)$-categories are classically equivalent to the $(\infty,1)$-category of $\infty$-groupoids.
These models are not themselves democratic, so here we mean that their hearts $\heart{\mod{M}}$ present these $(\infty,1)$-categories in the above sense (\Cref{democracy}).
Typically, however, these models are analyzed by means of a \emph{Quillen model structure}, another form of presentation of an $(\infty,1)$-category, on $\mod{M}$ itself.
Such a structure is defined by three classes of maps: \emph{cofibrations}, \emph{weak equivalences}, and \emph{fibrations}.

\begin{definition}
  The \emph{Quillen model structure presented by $\mod{M} \in \Mod{\rmcMLTT}$}, if it exists, is the unique model structure on $\Ctx{\mod{M}}$ such that
  \begin{enumerate}[(a)]
  \item the fibrations are the retracts in $\Ctx{\mod{M}}^\to$ of context extensions, \ie, of morphisms $p_A \colon \Gamma.A \to \Gamma$ arising as pullbacks in $\PSh{\Ctx{\mod{M}}}$ of $\mod{M}(\clTm)$;
  \item the unique map $0 \to \Gamma$ is a cofibration for all $\Gamma \in \Ctx{\mod{M}}$.
  \end{enumerate}
\end{definition}

The uniqueness follows from a result of Joyal \cite[Theorem 15.3.1]{riehl:14}.
A model structure on a category $\cat{E}$ induces a fibration category structure on the full subcategory of $X \in \cat{E}$ such that $X \to 1$ is a fibration; for a model structure presented by $\mod{M} \in \Mod{\rmcMLTT}$, this is exactly $\Ctx{\heart{\mod{M}}}$, and the induced fibration category is exactly Kapulkin and Lumsdaine's.

\Cref{reversal-conservativity} allows us to translate proofs in ``opaque'' cubical type theory with reversals into proofs that do not use reversals, which can then be interpreted in $\infty$-groupoids via the aforementioned models.
However, it does not allow us to translate proofs in strict cubical type theories.
Fortunately, we can also use the twist construction to directly construct models of strict cubical type theory with reversals in $\infty$-groupoids.
In fact, we can reuse existing model constructions of the kind pioneered by Orton and Pitts~\cite{orton-pitts:18} out of the box.

\subsection{Orton--Pitts models}

Orton and Pitts~\cite{orton-pitts:18} give an abstract description of Cohen, Coquand, Huber, and M\"ortberg's model of cubical type theory in De Morgan cubical sets~\cite{cohen-coquand-huber-mortberg:15}.
Abstracting from the case of cubical sets, they fix a topos $\cat{E}$ equipped with an interval object $I$ and a suitable subobject $\OPCof \mono \Omega$ of the subobject classifier and isolate axioms on this data sufficient to construct a model of a strict cubical type theory in $\cat{E}$ where the interval is interpreted by $I$ and the cofibrations by $\OPCof$.
They assume that the interval $I$ has connections, but Angiuli, Brunerie, Coquand, Harper, Favonia, and Licata (ABCHFL)~\cite{abcfhl:21} subsequently gave a similar construction for intervals without such structure.
Extracting what we need from their main result and rephrasing in our language, we have:

\begin{proposition}[{{\cite[Theorem 2]{abcfhl:21}}}]
  \label{abcfhl-theorem}
  Let $\cat{E} = \PSh{\cat{C}}$ be a presheaf category on a finite product category $\cat{C}$, let $I \in \cat{E}$ be a representable object with distinct points $0,1 \colon 1 \to I$, and let $\OPCof \mono \Omegadec$ be a subobject of the levelwise decidable subobject classifier in $\PSh{\cat{C}}$ that classifies the diagonal $I \to I \times I$ and is closed under finite conjunction, finite disjunction, and universal quantification over $I$.
  Then there is a model $\mod{M}$ of $\rmcCTTs$ such that
  \begin{enumerate}[(a)]
  \item $\Ctx{\mod{M}} = \cat{E}$.
  \item $\mod{M}(\ccII) = \yo I \in \PSh{\cat{E}}$.
  \item the maps $p_A \colon \Gamma.A \to \Gamma$ arising as pullbacks in $\PSh{\cat{E}}$ of $\mod{M}(\clTm)$ are those equipped with a \emph{diagonal Kan composition structure} \cite[Definition 1]{abcfhl:21}.
  \end{enumerate}
  This model interprets the interval theory of all $f \colon I^n \to I$ in $\cat{C}$ and equations between them.
\end{proposition}

We call $(\cat{C},I,0,1,\OPCof)$ satisfying the conditions of \cref{abcfhl-theorem} an \emph{ABCHFL setup} and write $\mod{M}(\cat{C},0,1,I,\OPCof)$ for the resulting model.
The maps with diagonal Kan composition structure can be described by a simple lifting property.
Awodey proves the following for cartesian cubical sets, but the same proof applies in the setting of \Cref{abcfhl-theorem}.

\begin{proposition}[{{\cite[Proposition 4.15(2)$\Leftrightarrow$(3)]{awodey:26}}}]
  \label{abcfhl-fibrations}
  Let $(\cat{C},I,0,1,\OPCof)$ be an \emph{ABCHFL setup}.
  A morphism $f \colon Y \to X$ admits a diagonal Kan composition structure if and only if it has the right lifting property against the unique dashed map
  \[
    \begin{tikzcd}[row sep=small]
      A \ar{d}[left]{m} \ar{r}{\pair{\id}{zm}} &[1em] A \times I \ar{d} \ar[bend left=20]{ddr}{m \times I} \\
      B \ar{r} \ar[bend right=20]{drr}[below left,pos=0.2]{(\id,z)} & \pushout \bullet \ar[dashed]{dr}[below left]{\gencof{m}{z}} \\
      & & B \times I
    \end{tikzcd}
  \]
  for every $m \colon A \mono B$ classified by $\OPCof$ and $z \colon B \to I$.
\end{proposition}

We call a map an \emph{$(I,\OPCof)$-fibration} when it satisfies the property in \Cref{abcfhl-fibrations}.
As a lifting property, it is closed under retracts \cite[Lemma 11.1.4]{riehl:14}.

\begin{proposition}
  \label{twist-setup}
  If $(\cat{C},I,0,1,\OPCof)$ is an ABCHFL setup, then $(\cat{C},I \times I,r,s,\OPCof)$ is an ABCHFL setup for every $r \neq s \colon 1 \to I \times I$.
  Moreover, the classes of $(I \times I,\OPCof)$- and $(I,\OPCof)$-fibrations coincide.
\end{proposition}
\begin{proof}
  Because $\cat{C}$ is a finite product category by assumption, $I \times I$ is also representable.
  The diagonal $\Delta_{I \times I} \colon I \times I \to (I \times I) \times (I \times I)$ is the conjunction of $\subst{(\pi_0 \times \pi_0)}\Delta_I$ and $\subst{(\pi_1 \times \pi_1)}\Delta_I$, the pullbacks of $\Delta_I \colon I \to I \times I$ along the projections $\pi_0 \times \pi_0, \pi_1 \times \pi_1 \colon (I \times I) \times (I \times I) \to I \times I$.
  Thus it is classified by $\OPCof$.
  Universal quantification of $I \times I$ is iterated universal quantification over $I$, so $\OPCof$ is closed under this operation.
  This completes the proof of the first claim.

  For the second claim, recall that the class of maps with the left lifting property against a map is closed under retracts, composition, and pushouts along arbitrary maps \cite[Lemma 11.1.4]{riehl:14}.
  Given $m \colon A \mono B$ classified by $\OPCof$ and $z \colon B \to I$, we can write $\gencof{m}{z}$ as a retract of $\gencof{m}{z \times 0}$ where $z \times 0 \colon B \to I \times I$:
  \[
    \begin{tikzcd}
      B \sqcup_A (A \times I) \ar[mono]{d}[left]{\gencof{m}{z}} \ar{r} & B \sqcup_A (A \times (I \times I)) \ar[mono]{d}{\gencof{m}{z \times 0}} \ar{r} & B \sqcup_A (A \times I) \ar[mono]{d}{\gencof{m}{z}} \\
      B \times I \ar{r}[below]{B \times (I \times 0)} & B \times (I \times I) \ar{r}[below]{B \times \pi_0} & B \times I \rlap{.}
    \end{tikzcd}
  \]
  Thus every $(I \times I,\OPCof)$-fibration is an $(I,\OPCof)$-fibration.
  For the converse, given $m \colon A \mono B$ classified by $\OPCof$ and $z \colon B \to I \times I$, we read off the commutative diagram
  \[
    \begin{tikzcd}
      A \ar{d}[left]{m} \ar{r}{(\id,\pi_0zm)} & A \times I \ar{r}{(\id,\pi_1zm\pi_0)} \ar{d} &[-2em] (A \times I) \times I \ar{d} \ar[bend left=50]{dddr}[pos=0.3]{(m \times I) \times I} &[-2em] \\
      B \ar[bend right]{dr}[below left]{(\id, \pi_0z)} \ar{r} & \pushout B \sqcup_A (A \times I) \ar[dashed]{d}{\gencof{m}{\pi_0z}} \ar{r} & \pushout B \sqcup_A (A \times I) \times I) \ar{d} \ar[dashed,bend left=20]{ddr}[pos=.3]{\gencof{m}{z}} \\
      & B \times I \ar[bend right=15]{drr}[below left,pos=.3]{(\id,\pi_1z\pi_0)} \ar{r} & \pushout (B \times I) \sqcup_{A \times I} (B \times I) \times I \ar[dashed]{dr}[below left]{\gencof{(m \times I)}{\pi_1z\pi_1}} \\[-1em]
      & & & (B \times I) \times I
    \end{tikzcd}
  \]
  that $\gencof{m}{z}$ is a composite of a pushout of $\gencof{m}{\pi_0z}$ and $\gencof{(m \times I)}{\pi_1z\pi_1}$.
  This semantic counterpart to \Cref{twist-filling} shows that all $(I,\OPCof)$-fibrations are $(I \times I,\OPCof)$-fibrations.
\end{proof}

Using \Cref{twist-setup}, we can turn any ABCHFL model that presents $\infty$-groupoids into an ABCHFL model with reversals that also presents $\infty$-groupoids.
A suitable input is the category of presheaves on the \emph{semilattice cube category} $\square_{\lor}$, \ie, the cartesian cube category with one connection.
This is the algebraic theory generated by an object $I$ with points $0,1 \colon 1 \to I$ and a connection $\lor \colon I \times I \to I$ satisfying the axioms of a bounded join-semilattice.

\begin{proposition}[{{\cite[Theorems 4.34 \& 7.8]{cavallo-sattler:25}}}]
  \label{one-connection-model}
  The ABCHFL model $\mod{M}(\square_{\lor},I,0,1,\Omegadec)$ presents a Quillen model structure.
  Assuming classical logic, this model structure presents the $(\infty,1)$-category of $\infty$-groupoids.
\end{proposition}

\begin{theorem}
  The ABCHFL model $\mod{M}(\square_{\lor},I \times I,(0,1),(1,0),\Omegadec)$ interprets strict cubical type theory with reversals and presents a Quillen model structure.
  Assuming classical logic, this model structure presents the $(\infty,1)$-category of $\infty$-groupoids.
\end{theorem}
\begin{proof}
  \Cref{twist-setup,abcfhl-theorem} give the existence of the model of type theory.
  The existence of the model structure can be established by following Awodey's construction \cite{awodey:26}, but he considers only the cartesian cube category with its canonical interval object.
  The necessary components appear in generality in Awodey et al.~\cite[Lemma 3.7.2, Proposition 3.7.3]{accrs:24} (compare \cite[Theorem 4.4.9]{accrs:24}).
  By \Cref{twist-setup} and the uniqueness of the model structure presented by a model of type theory, this model structure is the same as that presented by $\mod{M}(\square_{\lor},I,0,1,\Omegadec)$, so \Cref{one-connection-model} proves the final claim.
\end{proof}

Although the original interval $I$ in $\square_{\lor}$ has a connection, this is not the case for $I \times I$ with the endpoints $(0,1)$ and $(1,0)$: in order to give the definition $(i_0,i_1) \lor (j_0,j_1) \defeq (i_0 \lor j_0, i_1 \land j_1)$ from \S\ref{sec:contributions}, we need \emph{both} connections in the base model.
Thus we only model cubical type theory with a reversal, not with a connection.
We expect that we can apply the procedure in this section to the second author's model of cubical type theory with two connections \cite{sattler:25}, even though it is not an ABCHFL model, but we leave this to future work.

\bibliography{refs}

\end{document}